\documentclass[journal]{IEEEtran}

\usepackage[caption=false,font=normalsize,labelfont=sf,textfont=sf]{subfig}
\usepackage{textcomp}
\usepackage{stfloats}
\usepackage{verbatim}

\usepackage{cite}
\usepackage{graphicx}
\usepackage{amsthm}
\usepackage{amsmath,amsfonts}
\usepackage{mathtools}
\usepackage{url}
\usepackage{glossaries}
\usepackage{algorithmic}
\usepackage{algorithm}
\usepackage{array}

\newtheorem{theorem}{Theorem}

\newtheorem{proposition}[theorem]{Proposition}

\newtheorem{definition}[theorem]{Definition}
\newtheorem{remark}{Remark}

\usepackage{xcolor}

\newcommand{\F}{\mathbb{F}}

\newcommand{\R}{\mathbb{R}}

\newcommand{\E}{\mathbb{E}}

\newcommand{\Prob}{\mathbb{P}}

\newcommand{\Loss}{\mathcal{L}}

\newcommand{\BLER}{\mathrm{BLER}}

\newcommand{\pb}{POWERBLAST}

\newacronym{awgn}{AWGN}{additive white Gaussian noise}
\newacronym{awgnf}{AWGN-F}{additive white Gaussian noise channel with feedback}
\newacronym{snr}{SNR}{signal-to-noise ratio}
\newacronym{bler}{BLER}{block error rate}
\newacronym{ber}{BER}{bit error rate}
\newacronym{sk}{SK}{Schalkwijk-Kailath}
\newacronym{pm}{PM}{posterior matching}
\newacronym{llr}{LLR}{log-likelihood ratio}
\newacronym{ppv}{PPV}{Polyanskiy, Poor, and Verd\'u}
\newacronym{dlfc}{DL-ECFCs}{deep-learned error-correcting feedback codes}

\newacronym{pam}{PAM}{pulse-amplitude modulation}
\newacronym{mvue}{MVUE}{minimum-variance unbiased estimator}
\newacronym{lmmse}{LMMSE}{linear minimum mean square error estimator}
\newacronym{mmse}{MMSE}{minimum mean square error estimator}
\newacronym{mse}{MSE}{mean square error}
\newacronym{dp}{DP}{dynamic programming}
\newacronym{map}{MAP}{maximum a posteriori}
\newacronym{cce}{CCE}{categorical cross-entropy}

\newacronym{fe}{FE}{feature extractor}

\newacronym{mlp}{MLP}{multilayer perceptron}

\newacronym{npmf}{NPMF}{Nonlinear Posterior-Mean Feedback Code}
\newacronym{apmf}{APMF}{Active Posterior-Mean Feedback Code}

\begin{document}

\title{Nonlinear Posterior-Mean Feedback Codes for AWGN Channels}


\author{Yingyao Zhou,~\IEEEmembership{Student Member,~IEEE,} Natasha Devroye,~\IEEEmembership{Fellow,~IEEE}, Milo\v{s}~\v{Z}efran, ~\IEEEmembership{Senior Member,~IEEE,} Gy\"{o}rgy~Tur\'{a}n
\thanks{This work was supported by NSF awards 2217023, 2240532 and by the AI National Laboratory Program (RRF-2.3.1-21-2022-00004). The contents of this article are solely the responsibility of the authors and do not necessarily represent the official views of the NSF.}
\thanks{Yingyao Zhou, Natasha Devroye and Milo\v{s} \v{Z}efran are with the Department of Electrical and Computer Engineering, University of Illinois Chicago, Chicago, IL, USA {(e-mail: \{yzhou238, devroye, mzefran\}@uic.edu)}.}
\thanks{Gy\"{o}rgy~Tur\'{a}n is with the Department of Mathematics, Statistics, and Computer Science, University of Illinois Chicago, Chicago, IL, USA 
{and with the HUN-REN-SZTE Research Group on AI, University of Szeged, Hungary (e-mail: gyt@uic.edu).} }}




\maketitle

\begin{abstract}
Feedback can improve the reliability of communication over additive white Gaussian noise (AWGN) channels. Classical feedback codes are interpretable, or easy to understand, but often rely on linear estimation, while deep-learned feedback codes can achieve strong performance but require many learned parameters and are often black boxes that are difficult to interpret. In this work, we propose a new posterior-mean feedback coding framework for AWGN channels with feedback. The proposed scheme uses posterior-mean refinement to construct nonlinear feedback codes under both noiseless passive feedback and noisy active feedback, with maximum a posteriori (MAP) decoding at the receiver. This scheme is analytically described, but a few learned parameters related to power allocation, and hence we consider this an ``interpretable,'' or human-understandable code. We further develop a projection-based design to improve robustness under noisy feedback and support larger message sizes. Numerical results show that the proposed schemes achieve strong finite-blocklength performance and outperform several analytical and learned feedback coding baselines, while using only a small number of learned design parameters.
\end{abstract}

\begin{IEEEkeywords}
AWGN channel with feedback, nonlinear feedback codes, posterior mean estimation, passive feedback, active feedback
\end{IEEEkeywords}



\section{Introduction}
%
%
%
%
\IEEEPARstart{F}{eedback} is widely used in modern communication systems to improve reliability and adapt transmission to the channel. Feedback appears in several forms, including automatic repeat request (ARQ), hybrid ARQ (HARQ), channel state information (CSI), and channel output feedback. ARQ and HARQ allow the receiver to request retransmission when decoding fails, while CSI feedback enables adaptation to varying channel conditions. In this paper, we focus on {non-time-varying channels with channel output feedback}, a setting initially studied by Shannon~\cite{shannon1956zero}, where the transmitter obtains information about the receiver's channel observations. {This setting provides an upper benchmark for schemes that process the receiver's observations.} Although feedback does not increase the capacity of memoryless point-to-point channels~\cite{shannon1956zero}, it can simplify coding schemes and improve reliability in the finite-blocklength regime. This regime is increasingly important for emerging 6G wireless systems, which are expected to support massive connectivity, low latency, and highly reliable short-packet communication, where reducing the block error rate at a fixed blocklength is often more meaningful than asymptotically approaching capacity~\cite{Chowdhury2020,Brinton2025,Durisi2016}.

For the AWGN channel with output feedback, several analytical feedback coding schemes have been developed. The \gls{sk} scheme is a classical linear feedback code for noiseless feedback and achieves capacity with doubly exponential block error decay in the blocklength~\cite{schalkwijk1966coding}. Posterior matching provides a general feedback coding framework for memoryless channels and recovers the \gls{sk} scheme for the AWGN channel as a special case~\cite{shayevitz2011optimal}. The \gls{ppv} scheme provides a nonlinear feedback code for single-bit transmission and asymptotically achieves zero error probability with minimum energy~\cite{PPV}. However, these schemes require noiseless feedback, and their performance can degrade significantly when the feedback link is noisy. For noisy passive feedback, Chance and Love proposed a concatenated coding scheme, known as the CL scheme, in which a linear feedback scheme is used as the inner code and a forward error-correcting code is used as the outer code~\cite{chance2011concatenated}. The inner linear scheme optimizes the encoding and combining coefficients to improve the received \gls{snr}. More recently, Mishra \emph{et al.} introduced a sequential linear coding framework and derived the optimal scheme within this class using dynamic programming, which outperforms the CL inner scheme in noisy feedback settings~\cite{mishra2023linear}. However, these approaches remain restricted to linear feedback codes. This restriction is fundamental, since linear feedback codes cannot achieve a positive code rate over Gaussian channels with noisy feedback~\cite{kim2007gaussian}. To address this, Ben-Yishai and Shayevitz introduced Modulo-SK, an active feedback scheme in which the receiver encodes its observation before transmitting it over the noisy feedback link using a modulo operation~\cite{ben2017interactive}. While Modulo-SK improves reliability, its performance gains rely on the asymptotic regime, where modulo-aliasing errors can be made negligible. The design of nonlinear feedback codes that directly minimize block error rate at finite blocklengths remains an open problem.

This challenge has motivated deep-learned feedback codes, starting from recurrent neural network (RNN)-based schemes such as Deepcode~\cite{kim2018deepcode, mashhadi2021drf} and later transformer-based designs such as GBAF~\cite{shao2023attentioncode,ozfatura2022all,kim2023robust}. More recently, Lightcode proposed a lightweight feedback code that achieves current state-of-the-art finite-blocklength reliability with significantly reduced model complexity~\cite{ankireddy2025lightcode}. Beyond passive feedback, deep learning has also been extended to active feedback settings through schemes such as ActiveFB and BAAF~\cite{activeFB, baaf}. Subsequently, to relax the fixed-blocklength constraint of existing learned feedback codes, Deep Variable-Length Feedback (DeepVLF) codes have been proposed to dynamically adapt the transmission length based on feedback, improving efficiency especially in high-rate regimes~\cite{ding2026deep}. {However, DeepVLF operates in a variable-length regime and is therefore not directly comparable with the fixed-blocklength setting considered in this work. Although we focus on fixed-blocklength coding, the proposed framework can be naturally extended to variable-length transmission using the explicitly computed posterior distribution to define a stopping rule.} Existing learned feedback codes often rely on many parameters and complex architectures, making it difficult to understand, or ``interpret'', how feedback is used for error correction. Our previous work on interpreting Deepcode showed that learned feedback encoders often use nonlinear structure to enable power-efficient correction rules~\cite{Deepcode1, Deepcode2, Deepcode3}. However, these interpretations remain approximations of learned codes and do not provide fully analytical coding schemes with rigorous information-theoretic characterizations.

Motivated by these observations, this paper proposes a posterior-mean feedback coding framework for AWGN channels with feedback. The main idea is to directly use the posterior distribution to construct nonlinear feedback codes through posterior-mean refinement. The resulting schemes have analytical and interpretable encoder
structures and allow for \gls{map} decoding, while requiring only a small number of learned parameters. {Here, interpretability refers to the explicit analytical forms of the
encoder and decoder, which reveal how feedback is used for error correction, while
learning is limited to power-allocation and projection parameters.}

The main contributions are:
\begin{itemize}
    \item We propose a \gls{npmf} scheme for noiseless passive feedback. The closed-form encoder transmits posterior-mean estimation errors,
while the receiver performs optimal \gls{map} decoding using the
accumulated posterior distribution. We show that \gls{npmf} reduces to the \gls{ppv} scheme for single-bit communication and naturally extends to larger message sizes. {Numerical results show that in the considered short-blocklength settings,
\gls{npmf} achieves the lowest \gls{bler} among the evaluated
analytical and learned baselines.}

    \item We extend the posterior-mean framework to noisy active feedback through the proposed \gls{apmf}, where the receiver feeds back the posterior mean rather than channel outputs, thereby reducing the transmitter--receiver
mismatch caused by feedback noise.
    
    \item We further develop a learned projection-based \gls{apmf} scheme that
mitigates the one-dimensional packing limitation of \gls{pam} and
improves robustness for larger message sizes under noisy feedback. {The encoder and optimal \gls{map} decoder remain analytical,
with only $TK+(2T-1)$ trainable projection and power-allocation
parameters, orders of magnitude fewer than those of existing deep-learned feedback
codes, while achieving competitive or better finite-blocklength
\gls{bler} performance in the considered settings.}

    \item We derive analytical union bounds for the proposed schemes and
validate their finite-blocklength performance through numerical results.
    
\end{itemize}

The rest of this paper is organized as follows. Section~\ref{sec:system_model} introduces the system model. Section~\ref{sec:related_work} reviews existing analytical and learned feedback coding schemes. Section~\ref{sec:npmf} presents the proposed \gls{npmf} scheme and its analysis. Section~\ref{sec:apmf} extends the framework to \gls{apmf}. Section~\ref{sec:conclusion} concludes the paper.

\subsection{Notation} 
Random variables are denoted by capital letters (e.g., $M$), their realizations by lowercase letters (e.g., $m$), and vectors by boldface letters (e.g., $\mathbf{M}$). Probability is represented by $\Prob(\cdot)$, expectation by $\E[\cdot]$, and entropy by $\mathbb{H}(\cdot)$. Unless specified, all logarithms use base~2. The symbol $\F_2$ denotes the finite field with elements 0 and 1, while $\R^n$ represents $n$-dimensional real vectors. The cumulative distribution function (c.d.f.) of a random variable $M$ is defined as $F_M(m) = \Prob(M \leq m)$, and the inverse c.d.f. is defined as $F_M^{-1}(t) = \inf \{ m \in \mathbb{R} : F_M(m) \geq t \}, \quad t \in (0, 1)$. Subscripts $i$ represent time indices, while the superscript of a vector indicates its length, e.g., $\mathbf{M}^K = (M_1, \ldots, M_K)$.  $Q(x)$ denotes the complementary cumulative distribution function of the standard normal distribution $\mathcal{N}(0,1)$, given by $Q(x) = \frac{1}{\sqrt{2\pi}} \int_x^\infty \exp\left(-\frac{u^2}{2}\right) \, du$.  An interval in $\mathbb{R}$ is denoted by $[x, y]$. $\{h_{i}\}_{i=1}^T$ denotes a sequence of functions.


\section{System Model}\label{sec:system_model}
\begin{figure}[!t]
    \centering
    \includegraphics[width=\columnwidth]{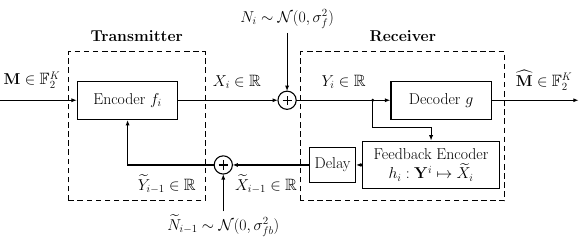}
    \caption{Point-to-point AWGN channel with feedback. In passive feedback, the receiver directly sends its past observations. In active feedback, the receiver encodes its past observations before sending them back.}
    \label{fig:awgn}
\end{figure}
We consider a real-valued point-to-point \gls{awgnf}, as illustrated in Fig.~\ref{fig:awgn}. 
The transmitter sends a message $\mathbf{M} = (M_1, \ldots,\allowbreak M_K)\in \mathbb{F}_2^K$ to the receiver over $T$ channel uses. 
The code rate is defined as $R=K/T$. At channel use $i\in\{1,2,\ldots,T\}$, the forward channel is modeled as
\begin{equation}
Y_i = X_i + N_i,
\end{equation}
where $N_i\sim \mathcal{N}(0,\sigma_f^2)$ denotes independent and identically distributed (i.i.d.) Gaussian noise, and $X_i \in \mathbb{R}$ is the transmitted symbol at time $i$, generated from the message $\mathbf{M}$ and the available past feedback. 

We assume a feedback link with one-unit delay. Thus, at time $i$, the transmitter receives feedback based on the previous channel outputs $\mathbf{Y}^{i-1} = (Y_1, \ldots, \allowbreak Y_{i-1})$. We consider both passive and active feedback. In passive feedback, the receiver directly sends the past channel output, 
\begin{equation}
    \widetilde{X}_{i-1}=Y_{i-1}, \qquad i\in\{2,\ldots,T\}.
\end{equation}
In active feedback, the receiver first encodes its past observations $\mathbf{Y}^{i-1}$ through a feedback encoder $h_{i-1}:\R^{i-1}\rightarrow\R$ and sends
\begin{equation}
    \widetilde{X}_{i-1}=h_{i-1}(\mathbf{Y}^{i-1}), \qquad i\in\{2,\ldots,T\}.
\end{equation}
The transmitter observes the noisy feedback output
\begin{equation}
    \widetilde{Y}_{i-1}
    =
    \widetilde{X}_{i-1}
    +
    \widetilde{N}_{i-1},
\end{equation}
where $\widetilde{N}_{i-1}\sim\mathcal{N}(0,\sigma_{fb}^2)$ is i.i.d. feedback noise, with $\widetilde{N}_{i-1}=0$ in the noiseless feedback case.



An $(2^K,T)$ feedback code consists of a sequence of encoding functions $\{f_i\}_{i=1}^T$, optional feedback encoding functions $\{h_i\}_{i=1}^{T-1}$, and a decoding function $g$. At channel use $i$, the transmitter maps the message $\mathbf{M}$ and past feedback observations $\widetilde{\mathbf{Y}}^{i-1}=(\widetilde{Y}_1,\ldots,\widetilde{Y}_{i-1})$ to the transmitted symbol $X_i\in\mathbb{R}$, 
\begin{equation}
    X_i = f_i(\mathbf{M},\widetilde{\mathbf{Y}}^{i-1}),
    \qquad i=1,\ldots,T .
\end{equation}

{The forward symbols satisfy either a per-symbol average power constraint,}
\begin{equation}
    \E[X_i^2]\leq P_f,
    \qquad i=1,\ldots,T,
\end{equation}
or a per-codeword average power constraint,
\begin{equation}
    \frac{1}{T}\E\!\left[\sum_{i=1}^T X_i^2\right]\leq P_f.
    \label{eq:avg_power_constraint}
\end{equation}
Similarly, for active feedback, the feedback symbols satisfy the corresponding per-symbol or per-codeword average power constraint. For the per-codeword case,
\begin{equation}
    \frac{1}{T-1}\E\!\left[\sum_{i=1}^{T-1}\widetilde{X}_i^2\right]\leq P_{fb}
\end{equation}
for the per-codeword average-power case.

\begin{remark}
For passive feedback, the transmission power is determined by the forward channel statistics and is bounded by $P_f + \sigma_f^2$. Therefore, following prior work (e.g., Deepcode~\cite{kim2018deepcode}), we do not impose an explicit power constraint on passive feedback. We therefore evaluate passive and active feedback schemes relative to their corresponding baselines.
\end{remark}

The forward and feedback \glspl{snr} are defined as
\begin{equation}
    \mathrm{SNR}_f \coloneqq \frac{P_f}{\sigma_f^2},
    \qquad
    \mathrm{SNR}_{fb} \coloneqq \frac{P_{fb}}{\sigma_{fb}^2}.
\end{equation}
Without loss of generality, we set $P_f=1$ and, for active feedback, $P_{fb}=1$.
After $T$ channel uses, the decoder $g:\mathbb{R}^T\rightarrow\mathbb{F}_2^K$ maps the received outputs to an estimate of the transmitted message,
\begin{equation}
\widehat{\mathbf{M}} = g(\mathbf{Y}^T).
\end{equation}

The reliability of communication is measured by the \gls{bler}, defined as
\begin{equation}
 \BLER \coloneqq \mathbb{P}(\mathbf{M} \neq \widehat{\mathbf{M}}).
\end{equation}
The forward encoder, feedback encoder, and decoder may be analytically designed or learned from data. 
In this paper, we consider passive feedback for  noiseless feedback and active feedback for noisy feedback. In the following sections, we first review existing analytical and neural feedback coding schemes, and then introduce the proposed posterior-mean feedback codes.


\section{Existing Feedback Coding Schemes}\label{sec:related_work}

In this section, we review existing feedback coding schemes for the \gls{awgnf}. We first discuss schemes under passive feedback, including the analytical \gls{sk} scheme~\cite{schalkwijk1966coding}, its variant {\pb}~\cite{ankireddy2025lightcode}, the \gls{pm} scheme~\cite{shayevitz2011optimal}, and the nonlinear \gls{ppv} scheme based on the \gls{llr}~\cite{PPV}. We also discuss deep-learned feedback codes under passive feedback, where Lightcode~\cite{ankireddy2025lightcode} represents the current state of the art~\cite{kim2018deepcode,shao2023attentioncode,ozfatura2022all,kim2023robust,ankireddy2025lightcode}.

We then review schemes under active feedback, including Modulo-\gls{sk}~\cite{ben2017interactive}, BAAF~\cite{baaf}, and activeFB~\cite{activeFB}. These schemes use a feedback encoder at the receiver to improve robustness when the feedback link is noisy.

\subsection{Passive Feedback Schemes}
\subsubsection{Schalkwijk-Kailath (SK) Scheme}

The \gls{sk} scheme is a linear feedback code designed for noiseless output feedback~\cite{schalkwijk1966coding}. Initially, a message $\mathbf{M} \in \mathbb{F}_2^K$ is mapped to a \gls{pam} symbol $\Theta \in \{\pm \eta, \pm 3\eta, \dots, \pm (2^K-1)\eta\}$, where $\eta = \sqrt{3/(2^{2K}-1)}$ {gives $\E[\Theta^2]=1$}. In the first round, the transmitter sends with power $P_f$, $X_1 = \sqrt{P_f}\Theta$. The estimated \gls{pam} symbol at the receiver at time $i$ is denoted by $\widehat{\Theta}_i$, with the corresponding estimation error defined as $\epsilon_i = \widehat{\Theta}_i - \Theta$ and root mean square error as $D_i = \sqrt{\mathbb{E}[\epsilon_i^2]}$. After the first round, the initial estimate $\widehat{\Theta}_1$ is formed based on $Y_1$ using either the \gls{mvue}, $\widehat{\Theta}_1 = Y_1/\sqrt{P_f}$, or the \gls{lmmse}, $\widehat{\Theta}_1 =\sqrt{P_f}Y_1/(P_f+\sigma_f^2)$. In each subsequent round $i \in \{2, \dots, T\}$, the transmitter sends a scaled version of the previous estimation error,
\begin{equation*}
X_i = \frac{\sqrt{P_f}}{D_{i-1}}\epsilon_{i-1}
= \frac{\sqrt{P_f}}{D_{i-1}}\left(\widehat{\Theta}_{i-1}-\Theta\right).
\end{equation*}
The receiver then estimates this error using \gls{lmmse},
\begin{equation*}
    \widehat{\epsilon}_{i-1}
    = \frac{\sqrt{P_f}D_{i-1}}{P_f+\sigma_f^2}Y_i 
\end{equation*}
and updates its estimate recursively as $\widehat{\Theta}_i = \widehat{\Theta}_{i-1} - \widehat{\epsilon}_{i-1}$. Finally, the receiver maps the final estimate $\widehat{\Theta}_T$ to the nearest \gls{pam} constellation point.


We use the \gls{lmmse} estimator in the first round for finite-blocklength simulations due to its improved \gls{bler} performance, while adopting the more analytically tractable \gls{mvue} for the error probability analysis below. For rate $R=K/T$, the \gls{bler} under the \gls{mvue} and a per-symbol average power constraint is
\begin{equation*}
    \BLER_{\text{sk}} = 2 (1 - 2^{-K}) Q\left(\sqrt{\frac{3\mathrm{SNR}_f(1 + \mathrm{SNR}_f)^{T-1}}{2^{2K} - 1}}\right).
\label{eq:skpe}
\end{equation*}
This expression reflects the well-known doubly exponential error decay of the \gls{sk} scheme in the block length $T$.

The \pb~\cite{ankireddy2025lightcode} scheme is a variation of the
\gls{sk} scheme that differs only in the final transmission round {and uses the per-symbol average power constraint.} In this round, rather than sending the continuous estimation error,
the transmitter sends a discrete symbol representing the index error,
$X_T=\sqrt{P_f}[I(\widehat{\Theta}_{T-1})-I(\Theta)]/D_I$, where
$I(\cdot)\in\{1,\dots,2^K\}$ denotes the \gls{pam} symbol index and
$D_I=\sqrt{\mathbb{E}[(I(\widehat{\Theta}_{T-1})-I(\Theta))^2]}$.
At the receiver, a \gls{map} decoder is used to recover the original
message. Thus, the \gls{bler} for rate $R=K/T$ is
\begin{equation*}
\begin{aligned}
\BLER_{\text{pb}}
={}&
2(1-p_{\text{sk}})
Q\!\left(\gamma\sqrt{\mathrm{SNR}_f}\right) \\
&+
p_{\text{sk}}
Q\!\left(
\left(\frac{1}{p_{\text{sk}}}-\gamma\right)
\sqrt{\mathrm{SNR}_f}
\right),
\end{aligned}
\end{equation*}
where $p_{\text{sk}}$ is the $\BLER_{\text{sk}}$ for code rate
$R=K/(T-1)$, and $\gamma$ denotes the detection threshold
in~\cite{ankireddy2025lightcode}.

\subsubsection{Posterior Matching Scheme}
The posterior matching scheme~\cite{shayevitz2011optimal} is a general recursive feedback framework that extends the \gls{sk} to arbitrary memoryless channels. The main idea is to maintain a posterior distribution of the message at the receiver and transmit a symbol that makes each new observation as informative as possible. According to the posterior matching principle, the transmitted symbol must be a function of a random variable $V$ that captures the information about the message not contained in the past channel outputs $\mathbf{Y}^{i-1}$. Therefore, $V$ is statistically independent of $\mathbf{Y}^{i-1}$, ensuring that the transmission conveys only new information. In addition, the message point must be almost surely uniquely recovered from $(V,\mathbf{Y}^{i-1})$.

Specifically, let $\Theta_0 \sim \mathcal{U}(0,1)$ be the message point. The encoding function at the $i$-th round is 
\begin{equation}
X_i
=
F_X^{-1}
\!\left(
F_{\Theta_0|\mathbf{Y}^{i-1}}
(\Theta_0 \mid \mathbf{Y}^{i-1})
\right)
\label{eq:pm_enc}
\end{equation}
where $F_{\Theta_0|\mathbf{Y}^{i-1}}(\cdot \mid \mathbf{Y}^{i-1})$ denotes the posterior c.d.f. of $\Theta_0$ given $\mathbf{Y}^{i-1}$, and $F_X^{-1}$ is the inverse c.d.f. of the desired input distribution. 
{For the \gls{awgnf}, the capacity-achieving input distribution is
$X \sim \mathcal{N}(0,P_f)$, which satisfies the per-symbol average
power constraint $\E[X_i^2] = P_f$. Its c.d.f. is used as $F_X$
in \eqref{eq:pm_enc}.} In this case, the posterior c.d.f. is shaped by $F_X^{-1}$ into a Gaussian channel input, and the posterior matching scheme reduces to the \gls{sk} scheme, where the recursive update becomes linear. {The relation between posterior matching and the proposed \gls{npmf} is discussed in Section~\ref{sec:relationwithPM}.}

\subsubsection{Polyanskiy, Poor, and Verdú (PPV) Scheme}
The \gls{ppv} scheme is a nonlinear  coding scheme for single-bit transmission~\cite{PPV} in which the transmitted symbol is a nonlinear function of the \gls{llr}, scaled to meet the power constraint.
Specifically, for a single bit $M\in \mathbb{F}_2$, with BPSK modulation $W = 2M-1\in \{-1, 1\}$, the transmitted symbol at the $i$-th round is
\begin{equation}
X_i = f_{i}(W, \mathbf{Y}^{i-1}) = \frac{W d_i}{1 + e^{W S_{i-1}}}, \qquad i=1,\ldots,T ,
\label{eq:ppv}
\end{equation}
where the \gls{llr} is given by $S_{i-1} = \log\frac{\mathbb{P}\left(W = +1\mid \mathbf{Y}^{i-1}\right)}{\mathbb{P}\left(W = -1\mid \mathbf{Y}^{i-1}\right)}$, $\mathbf{Y}^{i-1} = (Y_1,\ldots,Y_{i-1})$, and $d_i$ is a coefficient related to the power constraint, {which can be either the per-symbol or per-codeword average power constraint}. By setting $d_i = d$ for all $i$, the \gls{ppv} scheme has been shown to achieve zero error probability with the minimum energy per bit as the {number of} channel uses $T \rightarrow \infty$~\cite{PPV}. For finite block lengths, optimizing $d_i$ yields further BLER improvement using \gls{dp}~\cite{activeFB}. 

After $T$ rounds, the receiver decodes the message $\widehat{W}$ based on the final LLR, $S_T$. The $\BLER$ after $T$ rounds can be represented as a recursion function~\cite{activeFB}:
\begin{equation}
    \BLER_{\text{ppv}} = Q\left(\frac{\mu_T}{\sigma_T}\right),
\label{eq:ppvpe}
\end{equation}
where $\mu_N$ and $\sigma_N$
are updated recursively according to \eqref{eq:mu} and \eqref{eq:sigma},
\begin{align}
    \mu_i &= \mu_{i-1} + \frac{d_i^2}{2\sigma_f^2},
    \label{eq:mu}\\
    \sigma_i^2 &= \sigma_{i-1}^2 + \frac{d_i^2}{\sigma_f^2}.
    \label{eq:sigma}
\end{align}
By symmetry, it is sufficient to analyze $\Prob\left(S_T < 0\mid W = +1\right)$. Given $W = +1$, the \gls{llr} evolves as
\begin{equation*}
S_i = S_{i-1} + \frac{d_i^2}{2\sigma_f^2} + \frac{d_i}{\sigma_f^2}N_i,
\end{equation*}
where $N_i \sim \mathcal{N}(0,\sigma_f^2)$. Consequently, $S_T$ is Gaussian with mean $\mu_T$ and variance $\sigma_T^2$, which leads to the BLER expression in \eqref{eq:ppvpe}.

\subsubsection{Lightcode}
Lightcode~\cite{ankireddy2025lightcode} parameterizes the encoder and decoder using neural networks. Both the encoder and decoder consist of a feedforward feature extractor followed by a multilayer perceptron. At the $i$-th round, the encoder generates the transmitted symbol based on the message and the past feedback,
\begin{equation*}
    X_i
    =
    \alpha_i
    f_{i}
    \left(
    \mathbf{M},
    \widetilde{Y}_1,
    \ldots,
    \widetilde{Y}_{i-1}
    \right),
    \quad i=1,\ldots,T,
\label{eq:lightcode_enc}
\end{equation*}
where $\alpha_i$ is a learned scaling coefficient satisfying the per-codeword average power constraint $\frac{1}{T}\sum_{i=1}^T\alpha_i^2 = P_f$. After $T$ rounds, the decoder estimates the transmitted message as
\[
\widehat{\mathbf M}=g(\mathbf Y^T).
\]


The neural network is trained by minimizing the expected \gls{cce} loss. Treating each of the $2^K$ messages $\mathbf{M}$ as a distinct class, the loss function is defined as
\begin{equation}
\Loss_{\text{CCE}}
=
\mathbb{E}_{\mathbf{Y}^T}
\left[
-\sum_{\mathbf{m}}
\mathbb{P}(\mathbf{m}\mid\mathbf{Y}^T)
\log_2\mathbb{Q}(\mathbf{m}\mid\mathbf{Y}^T)
\right],
\label{eq:ce}
\end{equation}
where $\mathbb{P}(\mathbf{m}\mid\mathbf{Y}^T)$ is the true posterior probability and $\mathbb{Q}(\mathbf{m}\mid\mathbf{Y}^T)$ is the predicted probability of message $\mathbf{m}$ given the channel outputs $\mathbf{Y}^T$. This loss can be decomposed as
\begin{equation*}
\Loss_{\text{CCE}}
=
\mathbb{H}(\mathbf{M}\mid\mathbf{Y}^T)
+
\mathbb{E}_{\mathbf{Y}^T}\!\left[
D_{\mathrm{KL}}\!\left(
\mathbb{P}(\mathbf{M}\mid\mathbf{Y}^T)
\,\Vert\,
\mathbb{Q}(\mathbf{M}\mid\mathbf{Y}^T)
\right)
\right].
\end{equation*}
Since $D_{\mathrm{KL}}(\mathbb{P}\|\mathbb{Q}) \geq 0$,
the minimum over predicted distributions is attained when the predicted posterior matches
the true posterior. Hence,
\begin{equation}
\Loss_{\text{CCE}}^\star
=
\mathbb{H}(\mathbf{M}\mid\mathbf{Y}^T).
\label{eq:ce_min}
\end{equation}


The following proposition relates the optimal \gls{cce}
loss to the \gls{bler} of the \gls{map} decoder.

\begin{proposition}\label{prop:sandwich}
Let $\BLER = \mathbb{P}(\mathbf{M} \neq \widehat{\mathbf{M}})$ denote the block error probability of the \gls{map} decoder. The optimal \gls{cce} loss $\Loss_{\text{CCE}}^\star$ satisfies the sandwich inequality
\begin{equation}
2 \BLER \leq \Loss_{\text{CCE}}^\star \leq \mathbb{H}_2(\BLER) + \BLER \log_2(2^K - 1),
\end{equation}
where $\mathbb{H}_2(p) = -p \log_2 p - (1-p)\log_2(1-p)$ is the binary entropy function.
\end{proposition}

The proof is provided in Appendix~\ref{apx:sandwich}.
For $K=1$, the upper bound reduces to $\mathbb{H}_2(\BLER)$,
recovering the result in~\cite{2023isit}.
Thus, a small optimal \gls{cce} loss implies a small \gls{bler}.

\subsection{Active Feedback Schemes}
\subsubsection{Modulo-SK Scheme}
The Modulo-\gls{sk} scheme extends the \gls{sk} scheme to noisy active feedback~\cite{ben2017interactive}. The forward transmission follows the per-symbol average power constraint as in the \gls{sk} scheme. The classical \gls{sk} scheme is highly sensitive to feedback noise, since the feedback error can accumulate over multiple rounds and degrade the performance. Modulo-\gls{sk} addresses this issue by feeding back an amplified estimate through a modulo operation. Let $\mathbb{M}_d[\cdot]$ denote the modulo operation with interval length $d$, defined as $\mathbb{M}_d[x] \coloneqq x - d \cdot \text{round}\left(\frac{x}{d}\right)$, and let $V_{i-1}\sim\mathcal{U}[-d/2,d/2)$ be a shared dither known to both the transmitter and receiver. The receiver sends
\begin{equation*}
    \widetilde{X}_{i-1}
    =
    \mathbb{M}_d
    \left[
    \gamma_{i-1}\widehat{\Theta}_{i-1}
    +
    V_{i-1}
    \right],
\end{equation*}
where $\gamma_{i-1}$ is a feedback amplification factor and $d=\sqrt{12P_{fb}}$ is chosen to satisfy the per-symbol average power constraint on the feedback link. After observing $\widetilde{Y}_{i-1}=\widetilde{X}_{i-1}+\widetilde{N}_{i-1}$, the transmitter uses its knowledge of $\Theta$ to compute  
\[
    \widetilde{\epsilon}_{i-1}
    =
    \frac{1}{\gamma_{i-1}}
    \mathbb{M}_d
    \left[
    \widetilde{Y}_{i-1}
    -
    \gamma_{i-1}\Theta
    -
    V_{i-1}
    \right].
\]
If no modulo-aliasing occurs, then $\widetilde{\epsilon}_{i-1}
= \epsilon_{i-1}
+ \widetilde{N}_{i-1}/\gamma_{i-1}$.
    
Thus, compared with direct noisy feedback, Modulo-\gls{sk} scales the feedback noise by a factor of $1/\gamma_{i-1}$. The transmitter then sends a scaled version of this noisy estimation error in the next forward transmission, and the receiver uses it to refine its estimate of $\Theta$. The reliability of Modulo-\gls{sk} is determined by both the modulo-aliasing probability and the final \gls{pam} decision error.

\subsubsection{Deep-Learned Active Feedback Schemes}
Recent deep-learned active feedback schemes jointly optimize the forward encoder, feedback encoder, and decoder. 
ActiveFB~\cite{activeFB} uses recurrent neural networks and is
demonstrated for short messages with $K=1$ and $2$. BAAF~\cite{baaf} extends the GBAF architecture~\cite{ozfatura2022all} to active feedback for short messages (e.g., $K=3$), with the encoder, feedback encoder, and decoder all implemented using transformer-based structures. Both ActiveFB and BAAF use the per-codeword average power constraint for the forward and feedback transmissions. These schemes show that learning the feedback encoder can improve reliability under noisy feedback.

\section{Noiseless Feedback: Passive PAM-Based Nonlinear Posterior-Mean Feedback Code}\label{sec:npmf}
Most existing analytical feedback coding schemes are linear, such as \gls{sk}, while nonlinear schemes such as \gls{ppv} are restricted to single-bit transmission. Deep-learned feedback codes can achieve superior performance, but they rely on complex architectures, such as transformers~\cite{ozfatura2022all} and recurrent neural networks~\cite{kim2018deepcode,kim2023robust}, and often require many parameters. Even lightweight designs such as Lightcode~\cite{ankireddy2025lightcode} remain difficult to interpret.

In this section, we propose an analytical and interpretable nonlinear scheme under noiseless passive feedback, called the \gls{npmf}. We call this scheme interpretable as it is analytically defined and the error-correction mechanism is understood; but it has several learned parameters related to the allocation of power over multiple rounds.  The encoder transmits the estimation error based on the posterior mean at each round, while the decoder performs \gls{map} decoding. The proposed scheme reduces to \gls{ppv} for a single bit, i.e., $K=1$, and naturally extends it to multi-bit transmission for $K\geq 2$. We further analyze its finite-blocklength error probability and experimentally show that, under noiseless feedback, \gls{npmf} can outperform Lightcode without requiring a large number of parameters, while remaining interpretable.

\subsection{The Proposed Scheme}
The proposed \gls{npmf} follows the principle of \gls{sk}, where the encoder iteratively transmits the estimation error to the receiver, enabling successive refinement of the message estimate. Consider an \gls{awgnf} channel with noiseless output feedback. The message $\mathbf{M}\in\F_2^K$ is mapped to a \gls{pam} symbol $\Theta\in\{\theta_1,\dots,\theta_U\}$, where
\begin{equation}
    \theta_j = (2j - 1 - U)\eta, \quad j = 1,\dots,U,
    \label{eq:pam}
\end{equation}
with $U=2^K$ and $\eta=\sqrt{3/(2^{2K}-1)}$.

\begin{definition}[Nonlinear Posterior-Mean Feedback Code]
An $(K,T)$ \gls{npmf} code is defined over $T$ rounds. In the first round, the encoder transmits
\begin{equation}
    X_1 = \frac{\lambda_1}{D_0}\Theta,
    \qquad
    D_0=\sqrt{\E[\Theta^2]}=1.
    \label{eq:first_round}
\end{equation}
For $i=2,\dots,T$, the encoder transmits the scaled estimation error
\begin{equation*}
    X_i
    =
    \frac{\lambda_i}{D_{i-1}}
    \left(
    \Theta-\widehat{\Theta}_{i-1}
    \right),
    \quad
    D_{i-1}
    =
    \sqrt{
    \E\!\left[
    \left(
    \Theta-\widehat{\Theta}_{i-1}
    \right)^2
    \right]
    },
\end{equation*}
where $\lambda_i$ is the power allocation coefficient, and $D_{i-1}$ is the root mean square estimation error. The posterior mean, {i.e., the MMSE estimator}, is
\begin{equation}
    \widehat{\Theta}_{i-1}
    =
    \E[\Theta\mid\mathbf{Y}^{i-1}]
    =
    \sum_{j=1}^{U}
    \theta_j
    \Prob(\Theta=\theta_j\mid\mathbf{Y}^{i-1}).
\end{equation}
After $T$ rounds, the \gls{map} decoder produces
\begin{equation}
    \widehat{\Theta}
    =
    \arg\max_{\theta_j\in\{\theta_1,\dots,\theta_U\}}
    \Prob(\Theta=\theta_j\mid\mathbf{Y}^T),
\end{equation}
which is then mapped to $\widehat{\mathbf{M}}$.
\end{definition}

To compute the posterior distribution, the receiver maintains the log-likelihood of each symbol $j$ after $i$ channel uses,
\begin{equation}
    \ell_i(\theta_j)
    \coloneqq
    \log \Prob(\mathbf{Y}^i\mid \Theta=\theta_j).
\end{equation}
With forward noise $N_i\sim\mathcal{N}(0,\sigma_f^2)$, the log-likelihood updates recursively as
\begin{align*}
    \ell_i(\theta_j)
    &=
    \ell_{i-1}(\theta_j)
    -
    \frac{
    \left(
    Y_i-X_i(\theta_j)
    \right)^2
    }
    {2\sigma_f^2}
    -
    \frac{1}{2}\log(2\pi\sigma_f^2) \nonumber\\
    &=
    \ell_{i-1}(\theta_j)
    +
    \frac{1}{\sigma_f^2}
    \left(
    Y_iX_i(\theta_j)
    -
    \frac{X_i^2(\theta_j)}{2}
    \right)
    +
    C_i,
\end{align*}
where $C_i$ is independent of $\theta_j$ and $X_i(\theta_j)$ is the expected transmission symbol when $\Theta = \theta_j$. Since the $U$-ary \gls{pam} symbols are equiprobable, the posterior distribution is the softmax of the log-likelihoods
\begin{equation}
    \Prob(\Theta=\theta_j\mid\mathbf{Y}^i)
    =
    \frac{
    \exp(\ell_i(\theta_j))
    }
    {
    \sum_{k=1}^{U}\exp(\ell_i(\theta_k))
    }.
    \label{eq:absolute_softmax}
\end{equation}


To express the relative belief between two constellation symbols, we define the pairwise \gls{llr}
\begin{equation}
    S_i^{(j,k)}
    =
    \log
    \frac{
    \Prob(\Theta=\theta_j\mid \mathbf{Y}^i)
    }{
    \Prob(\Theta=\theta_k\mid \mathbf{Y}^i)
    }
    =
    \ell_i(\theta_j)-\ell_i(\theta_k).
\end{equation}
For $K=1$, this reduces to the binary \gls{llr} used in the \gls{ppv} scheme.


For any reference index $k$, the posterior mean can be written as
\begin{align}
\widehat{\Theta}_i
&=
\sum_{j=1}^U
\theta_j
\Prob(\Theta=\theta_j\mid \mathbf{Y}^i) \nonumber\\
&=
\frac{
\sum_{j=1}^U
\theta_j
\exp\!\left(S_i^{(j,k)}\right)
}{
\sum_{j=1}^U
\exp\!\left(S_i^{(j,k)}\right)
}.
\label{eq:mmse_exact}
\end{align}
This expression shows that $\widehat{\Theta}_i$ is a nonlinear, softmax-weighted function of the pairwise \gls{llr}s. In contrast to the linear MMSE update in \gls{sk}, the proposed encoder refines the message using the full posterior distribution, which leads to a more power-efficient structure discussed in Section~\ref{sec:skconnection}.

We next analyze the probability of error after $T$ channel uses. Suppose the true transmitted symbol is $\theta_j$. A decoding error occurs if there exists some $k\neq j$ such that 
\begin{equation}
    S_T^{(j,k)}
    =
    \ell_T(\theta_j)-\ell_T(\theta_k)
    <0 .
    \label{eq: error}
\end{equation}
For each round $i$, the transmitted-symbol difference between
hypotheses $\theta_j$ and $\theta_k$ is
\begin{equation*}
X_i(\theta_j)-X_i(\theta_k)
=
\frac{\lambda_i}{D_{i-1}}(\theta_j-\theta_k).
\end{equation*}
Therefore, the pairwise log-likelihood difference is
\begin{equation*}
\begin{aligned}
S_T^{(j,k)}
= \sum_{i=1}^{T} \Bigg[
&\frac{
(\lambda_i/D_{i-1})^2(\theta_j-\theta_k)^2
}{2\sigma_f^2} \\
&+
\frac{
(\lambda_i/D_{i-1})(\theta_j-\theta_k)
}{\sigma_f^2}N_i
\Bigg].
\end{aligned}
\end{equation*}
Since the \gls{pam} constellation spacing is $2\eta$, it follows that $\theta_j-\theta_k = 2\eta(j-k)$.
Consequently, $S_{T}^{(j,k)}$ is Gaussian distributed with mean $\mu_{j,k}$ and variance $\sigma_{j,k}^2$ given by:
\begin{align*}
    \mu_{j,k} & = \frac{2\eta^2(j-k)^2}{\sigma_f^2} \sum_{i=1}^T \frac{\lambda_i^2}{D_{i-1}^2}, \\
    \sigma_{j,k}^2 &= \frac{4\eta^2(j-k)^2}{\sigma_f^2} \sum_{i=1}^T \frac{\lambda_i^2}{D_{i-1}^2} = 2\mu_{j,k}.
\end{align*}

The error probability of choosing $\theta_k$ over $\theta_j$ is
\begin{equation}
    \Prob\!\left(S_T^{(j,k)}<0\right)
    =
    Q\!\left(
    \frac{\mu_{j,k}}{\sigma_{j,k}}
    \right)
    =
    Q\!\left(
    |j-k|\sqrt{\Lambda_T}
    \right),
\end{equation}
where
\begin{equation}
    \Lambda_T
    =
    \frac{\eta^2}{\sigma_f^2}
    \sum_{i=1}^{T}
    \frac{\lambda_i^2}{D_{i-1}^2}.
\end{equation}

Applying the union bound and averaging over the $U=2^K$ equiprobable \gls{pam} symbols, the overall BLER is bounded by
\begin{equation}
\label{eq:union_bound}
    \BLER_{\text{NPMF}}
    \leq
    \frac{2}{U}
    \sum_{m=1}^{U-1}
    (U-m)
    Q\!\left(
    m\sqrt{\Lambda_T}
    \right).
\end{equation}
Under the per-symbol average power constraint, $\lambda_i=1$ for all $i$. Under the per-codeword average power constraint, the power allocation coefficients $\{\lambda_i\}_{i=1}^T$ are optimized across rounds.

\subsection{Connections to Existing Schemes}
In this subsection, we establish the connection between the proposed \gls{npmf} scheme and classical feedback codes. We first show that \gls{npmf} reduces to \gls{ppv} when transmitting a single bit, and then compare it with other well-known linear and nonlinear coding schemes.

\subsubsection{PPV}
Both \gls{npmf} and \gls{ppv} are nonlinear feedback coding schemes, where the transmitted symbols are non-linear functions of both the message and the past feedback observations. The \gls{ppv} scheme is derived from a minimum-energy principle for single-bit transmission and is extended to multi-bit messages through time multiplexing~\cite{PPV}. 

In contrast, \gls{npmf} maps the entire $K$-bit message to a single \gls{pam} symbol and refines this symbol over all $T$ rounds using posterior-mean updates. Furthermore, the \gls{npmf} encoder is shown to minimize the transmission power $\E[|X_i|^2]$ at each time step $i$ subject to an equal-spacing constraint. The following proposition formalizes their connection.

\begin{proposition}
Consider the \gls{npmf} encoder
\begin{equation*}
    X_i
    =
    \frac{\lambda_i}{D_{i-1}}
    \left(
    \Theta-\E[\Theta\mid \mathbf{Y}^{i-1}]
    \right).
\end{equation*}

\begin{itemize}
\item[(i)] For $K=1$ with $\Theta\in\{-1,+1\}$, \gls{npmf} reduces to
\begin{equation*}
    X_i
    =
    \frac{2\lambda_i\Theta}
    {D_{i-1}\left(1+e^{\Theta S_{i-1}}\right)},
\end{equation*}
where $S_{i-1}$ is the binary \gls{llr}. Thus, \gls{npmf} coincides with the \gls{ppv} encoder with $d_i=2\lambda_i/D_{i-1}$.

\item[(ii)] Let $\Theta\in\{\theta_1,\dots,\theta_U\}$ be equiprobable.
Among encoders $X_i=f_i(\Theta,\mathbf{Y}^{i-1})$ satisfying, for
$j=2,\dots,U$, the equal-spacing constraint
\begin{equation*}
    f_i(\theta_j,\mathbf{Y}^{i-1})
    -
    f_i(\theta_{j-1},\mathbf{Y}^{i-1})
    =
    \frac{\lambda_i}{D_{i-1}}
    \left(
    \theta_j-\theta_{j-1}
    \right),
\end{equation*}
the unique minimizer of $\E[|X_i|^2]$ is
\begin{equation*}
    f_i^\star(\Theta,\mathbf{Y}^{i-1})
    =
    \frac{\lambda_i}{D_{i-1}}
    \left(
    \Theta-\E[\Theta\mid \mathbf{Y}^{i-1}]
    \right).
\end{equation*}
\end{itemize}
\end{proposition}

\begin{proof}
(i) For $K=1$, $\Theta \in \{-1,+1\}$ and $S_{i-1}$ denotes the binary \gls{llr}. The posterior mean in \eqref{eq:mmse_exact} is given by
\begin{equation*}
    \mathbb{E}[\Theta \mid \mathbf{Y}^{i-1}]
= \frac{e^{S_{i-1}} - 1}{e^{S_{i-1}} + 1}.
\end{equation*}
Substituting into $X_i = \frac{\lambda_i}{D_{i-1}}(\Theta - \widehat{\Theta}_{i-1})$ gives
\begin{align*}
    X_i=
\begin{cases}
\dfrac{2\lambda_i}{D_{i-1}(1+e^{S_{i-1}})}, & \Theta=+1, \\
-\dfrac{2\lambda_i}{D_{i-1}(1+e^{-S_{i-1}})}, & \Theta=-1.
\end{cases}
\end{align*}
Equivalently, letting $d_i = 2\lambda_i/D_{i-1}$ gives
\begin{equation*}
    X_i=\frac{d_i\Theta}{1+e^{\Theta S_{i-1}}},
\end{equation*}
which coincides with the \gls{ppv} encoder in \eqref{eq:ppv}.

(ii) Since $\Theta$ is equiprobable,
\begin{equation*}
    \E[|X_i|^2]
    =
    \frac{1}{U}
    \sum_{j=1}^U
    \E\!\left[
    |f_i(\theta_j,\mathbf{Y}^{i-1})|^2
    \mid
    \Theta=\theta_j
    \right].
\end{equation*}
Fix a reference symbol $\theta_k$. By a change of measure,
\begin{equation*}
\begin{aligned}
&\E\!\left[
|f_i(\theta_j,\mathbf{Y}^{i-1})|^2
\mid \Theta=\theta_j
\right] \\
&\qquad =
\E\!\left[
|f_i(\theta_j,\mathbf{Y}^{i-1})|^2
e^{S_{i-1}^{(j,k)}}
\mid \Theta=\theta_k
\right],
\end{aligned}
\end{equation*}
where
\begin{equation*}
    S_{i-1}^{(j,k)}
    =
    \log
    \frac{
    \Prob(\Theta=\theta_j\mid \mathbf{Y}^{i-1})
    }{
    \Prob(\Theta=\theta_k\mid \mathbf{Y}^{i-1})
    }.
\end{equation*}
Thus, for each realization of $\mathbf{Y}^{i-1}$, minimizing $\E[|X_i|^2]$ is equivalent to minimizing
\begin{equation}
\label{eq:minimum}
    J(\mathbf{Y}^{i-1})
    :=
    \sum_{j=1}^U
    |f_i(\theta_j,\mathbf{Y}^{i-1})|^2
    e^{S_{i-1}^{(j,k)}}.
\end{equation}
Iterating the equal-spacing constraint gives
\begin{equation*}
    f_i(\theta_j,\mathbf{Y}^{i-1})
    =
    \frac{\lambda_i}{D_{i-1}}\theta_j
    -
    h_i(\mathbf{Y}^{i-1}),
\end{equation*}
where $h_i(\mathbf{Y}^{i-1})$ is common to all constellation points. Substituting this form into \eqref{eq:minimum}, 
\begin{equation*}
    J(\mathbf{Y}^{i-1})
    =
    \sum_{j=1}^{U}
    \left(
    \frac{\lambda_i}{D_{i-1}}\,\theta_j
    -
    h_i(\mathbf{Y}^{i-1})
    \right)^{\!2}
    e^{S_{i-1}^{(j,k)}},
\end{equation*}
which is strictly convex in $h_i$. Setting $\partial J/\partial h_i = 0$ gives the unique minimizer
\begin{equation*}
    h_i^\star
    =
    \frac{\lambda_i}{D_{i-1}}\frac{
    \sum_{j=1}^U
    \theta_j e^{S_{i-1}^{(j,k)}}
    }{
    \sum_{j=1}^U
    e^{S_{i-1}^{(j,k)}}
    }
    =
    \frac{\lambda_i}{D_{i-1}}
    \E[\Theta\mid \mathbf{Y}^{i-1}],
\end{equation*}
which does not depend on the choice of reference $\theta_k$. We conclude
\begin{equation*}
    f_i^\star(\Theta,\mathbf{Y}^{i-1})
    =
    \frac{\lambda_i}{D_{i-1}}
    \left(
    \Theta-\E[\Theta\mid \mathbf{Y}^{i-1}]
    \right),
\end{equation*}
which is exactly the \gls{npmf} encoder.
\end{proof}

Although \gls{npmf} and \gls{ppv} are equivalent for $K=1$, their extensions to $K \ge 2$ differ fundamentally. To transmit a $K$-bit message over $T$ channel uses, the single-bit \gls{ppv} scheme can be extended through time multiplexing~\cite{PPV}, allocating roughly $T/K$ channel uses to each bit and refining each bit independently. In contrast, the proposed \gls{npmf} scheme maps the entire message to a single \gls{pam} symbol $\Theta$ and uses all $T$ rounds to iteratively refine it via the posterior mean. For a fixed code rate $R=K/T$, increasing $K$ increases $T$, allowing
\gls{npmf} to refine the entire message over a longer blocklength.
At high forward \gls{snr}, this leads to faster \gls{bler} decay than
the single-bit refinement strategy of the time-multiplexed \gls{ppv}
scheme, as shown in Fig.~\ref{fig:npmf_fixed_performance}.

\subsubsection{SK and {\pb}} \label{sec:skconnection}

{The \gls{sk} scheme uses a linear estimation-error refinement,
while \gls{npmf} uses the exact posterior mean of the discrete
\gls{pam} symbol, which is generally a nonlinear function of the
past channel observations.
Fig.~\ref{fig:scatter} shows the second-round encoder output $X_2$
versus the first-round noise $N_1$ for $K=3$ with an 8-ary
\gls{pam} constellation. For Lightcode, the mapping is approximately
linear for inner symbols, where noise in either direction can cause
confusion with neighboring symbols. At the boundary symbols
$\theta_1$ and $\theta_8$, favorable noise, e.g., $N_1>0$ for
$\theta_1$ or $N_1<0$ for $\theta_8$, pushes the observation further
into the correct decision region, so little correction is needed.
The proposed \gls{npmf} shows similar nonlinear behavior, but it
follows directly from the posterior mean. For example, when
$\Theta=\theta_1$, negative noise moves $Y_1$ further left, causing
the posterior to concentrate on $\theta_1$. Hence,
$\E[\Theta\mid Y_1]\approx\theta_1$ and the correction $X_2$
becomes small. When the noise moves $Y_1$ toward neighboring symbols,
the posterior remains uncertain and \gls{npmf} sends a stronger
correction. Thus, \gls{npmf} allocates more power to difficult
observations than to already reliable ones.}

\begin{figure}[!t]
    \centering
    \includegraphics[width=\columnwidth]{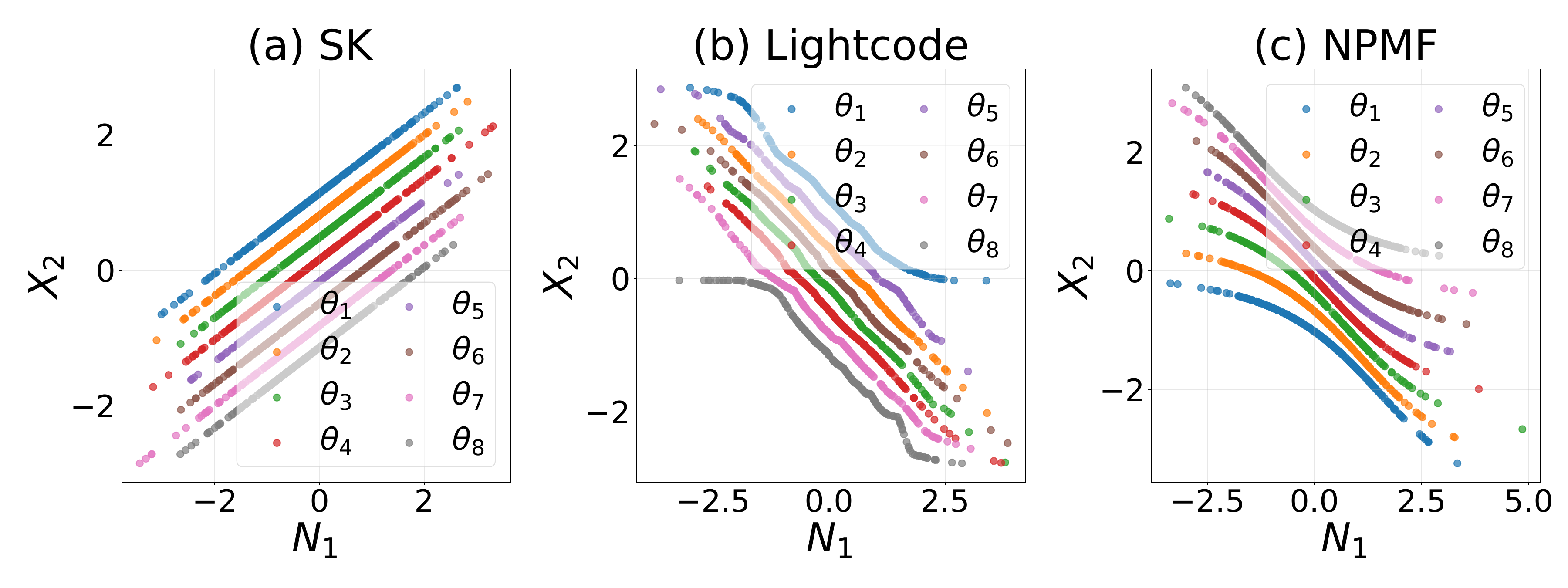}
    \caption{Second-round encoder output $X_2$ versus the first-round forward noise $N_1$ for $K=3$ (8-ary \gls{pam}). The nonlinear structure of \gls{npmf} and Lightcode contrasts with the linear behavior of \gls{sk}.}
    \label{fig:scatter}
\end{figure}
 
POWERBLAST is a variant of the \gls{sk} scheme. We also consider a modified \gls{npmf} scheme in which the final round is replaced by transmitting the PAM index difference, as in POWERBLAST. We refer to this variant as NPMF-PB.

\subsubsection{Posterior Matching}\label{sec:relationwithPM}
While both the proposed \gls{npmf} scheme and posterior matching use the posterior probability of the message, NPMF does not belong to the posterior matching class. Posterior matching encodes via a posterior c.d.f transformation as in \eqref{eq:pm_enc}, which guarantees that the transmitted symbol is statistically independent of $\mathbf{Y}^{i-1}$. In contrast, NPMF transmits the scaled posterior mean error $\Theta - \mathbb{E}[\Theta|\mathbf{Y}^{i-1}]$, which by the orthogonality principle of MMSE estimation is uncorrelated with $\mathbf{Y}^{i-1}$, but not necessarily statistically independent. Since $\Theta$ is uniformly distributed over a discrete PAM constellation, the posterior mean $\mathbb{E}[\Theta|\mathbf{Y}^{i-1}]$ is a nonlinear function of $\mathbf{Y}^{i-1}$, and uncorrelatedness does not imply statistical independence for non-Gaussian distributions. Hence, NPMF does not satisfy the statistical independence requirement of the posterior matching principle. {For the \gls{awgnf} channel, posterior matching reduces to the \gls{sk} scheme, which \gls{npmf} outperforms in the considered finite-blocklength simulations.}

\subsection{Simulation Results}
We evaluate the proposed \gls{npmf} scheme under both per-symbol and per-codeword average power constraints. The decoder computes the posterior probabilities using the softmax expression in \eqref{eq:absolute_softmax}, and the final posterior distribution is used in the \gls{cce} loss for optimization. The learnable parameters are optimized using AdamW
with a batch size of $10^5$, an initial learning rate of $10^{-3}$,
and weight decay of $10^{-2}$. Results are averaged over three
random seeds.\footnote{Codes are available at
\url{https://github.com/zyy-cc/npmf}.} Since the normalization factor $D_{i-1}$ depends on the power allocation across rounds, it is computed per batch during training. During testing, we fix $D_{i-1}$ using empirical estimates obtained from $10^6$ independent samples. For \gls{bler} evaluation, each operating point is simulated using
direct Monte Carlo simulation until 100 errors are observed.

As shown in Fig. \ref{fig:npmf_noiseless_performance}, we evaluate the proposed \gls{npmf} scheme against analytical baselines and the state-of-the-art Lightcode under noiseless feedback for rates $K/T = 3/6$ and $K/T = 3/9$, considering both per-symbol and per-codeword average power constraints. The \gls{npmf} under the per-codeword average power constraint consistently achieves the lowest BLER across the entire SNR range, outperforming all competing methods. The simulation results closely align with the analytical union bound in \eqref{eq:union_bound}. For $R = 3/6$, the hybrid NPMF-PB variant, which incorporates POWERBLAST principles, performs worse than the pure NPMF, highlighting the benefit of the posterior mean. For the lower rate $R = 3/9$, \gls{npmf} achieves strong reliability even at low SNR of $-1$ dB, outperforming both POWERBLAST and Lightcode. Moreover, learned codes such as Lightcode often reach a performance limit at higher SNR due to the scarcity of error events for gradient-based refinement, while the proposed \gls{npmf} remains effective under {a per-symbol average power constraint}. {At $K/T=3/9$, \gls{npmf} uses only 8 power-allocation parameters, compared with $7.3\times10^3$ parameters in Lightcode~\cite{ankireddy2025lightcode}.}

\begin{figure}[!t]
    \centering
    \includegraphics[width=0.9\columnwidth]{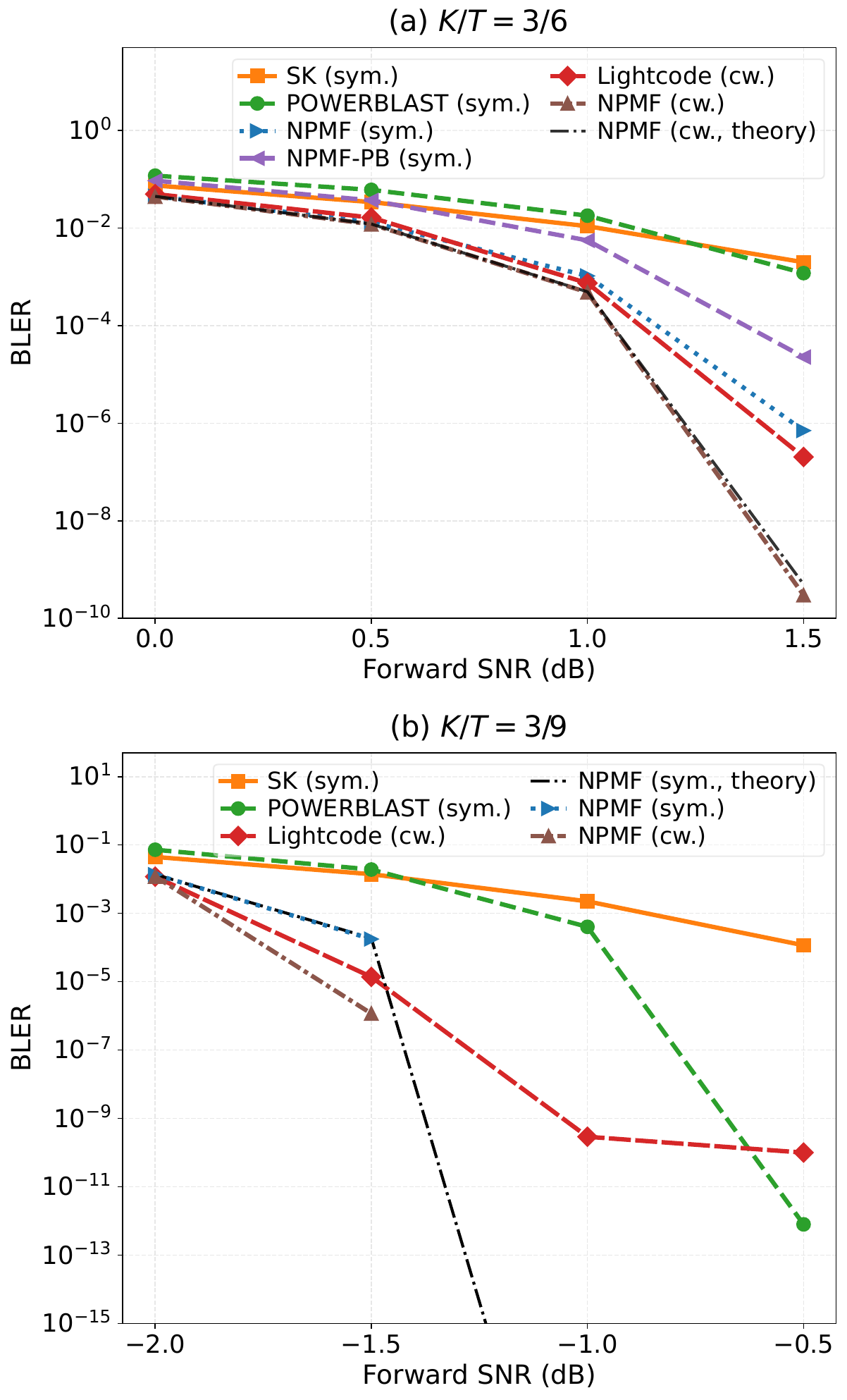}
    \caption{Performance of existing feedback codes for $K/T = 3/6$ and $K/T = 3/9$ under per-symbol (``sym.'') and per-codeword (``cw.'') average power constraints.}
    \label{fig:npmf_noiseless_performance}
\end{figure}

\begin{figure}[!t]
    \centering
    \includegraphics[width=0.9\columnwidth]{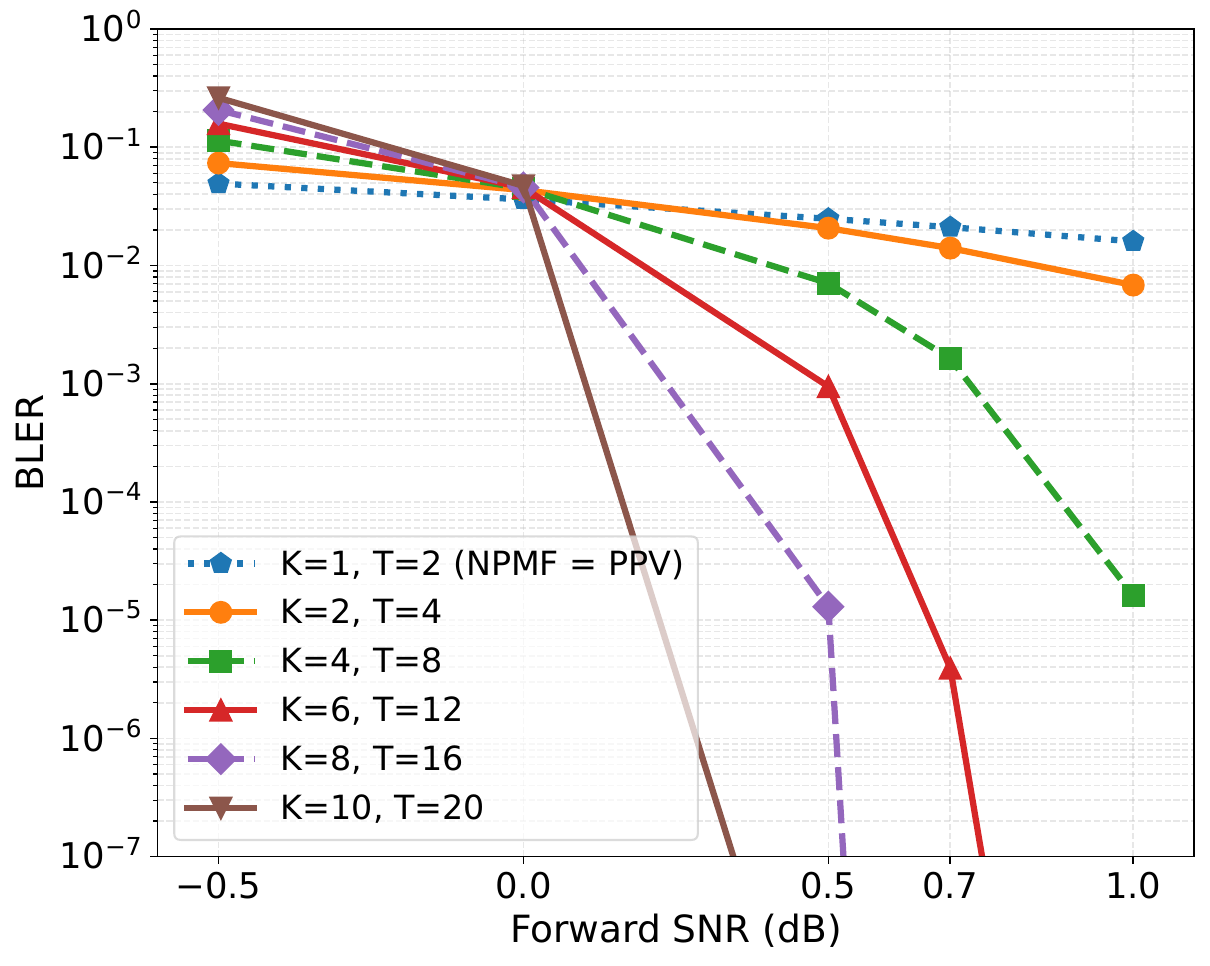}
    \caption{Performance of \gls{npmf} at a fixed code rate of $1/2$ under a per-symbol average power constraint}
    \label{fig:npmf_fixed_performance}
\end{figure}

We evaluate \gls{npmf} at a fixed rate $R=1/2$ under a per-symbol
average power constraint in Fig.~\ref{fig:npmf_fixed_performance}. Since the real-valued AWGN capacity equals $1/2$ at
$\mathrm{SNR}_f=0$ dB, reliable communication is not expected below
this threshold. At higher SNRs, larger $K$ benefits from the longer
blocklength $T$, leading to improved reliability. Unlike Lightcode, \gls{npmf} can support larger $K$ without a rapidly growing number of learned parameters.






\section{Noisy Feedback: Active Posterior-Mean Feedback Codes}\label{sec:apmf}
In this section, we extend the proposed posterior-mean feedback framework to noisy feedback. Although \gls{npmf} performs well with noiseless feedback, its performance degrades when feedback noise is present. The main difficulty is that the transmitter and receiver no longer maintain the same posterior state. Since the encoder relies on an accurate posterior mean, even small feedback noise can create a mismatch between the transmitter and receiver estimates. This mismatch is then amplified over multiple rounds and can lead to poor reliability.

To address this issue, we propose an \gls{apmf} scheme, where the receiver actively sends the posterior mean to the transmitter. We first introduce a \gls{pam}-based \gls{apmf} code. This scheme improves robustness under noisy feedback, but its performance degrades as $K$ increases because the $2^K$-ary \gls{pam} constellation points become more densely packed in one dimension. As a result, the code becomes more sensitive to feedback noise, even though more transmission rounds are available at a fixed rate. To overcome this limitation, we further propose a projection-based \gls{apmf} code. Instead of using classical lattice codes or constellation shaping to design complex high-dimensional constellations, the proposed scheme keeps a simple $K$-dimensional binary representation and learns a sequence of scalar projections. By using round-dependent learned projection directions, the encoder
forms different one-dimensional constellations across transmission
rounds to emphasize different message bits, thus mitigating the
one-dimensional packing limitation of \gls{pam} and improving
reliability under noisy active feedback.

\subsection{\gls{pam}-Based \gls{apmf} Code}
We first introduce the \gls{pam}-based \gls{apmf} code. As in \gls{npmf}, the message $\mathbf{M}\in\mathbb{F}_2^K$ is mapped to a \gls{pam} symbol $\Theta\in\{\theta_1,\dots,\theta_U\}$ defined in \eqref{eq:pam}. After each round, the receiver actively feeds back its posterior mean of $\Theta$, and the transmitter uses the noisy feedback observation to generate the next residual transmission.

\begin{definition}[\gls{pam}-Based Active Posterior-Mean Feedback Code]
An $(K,T)$ \gls{pam}-based \gls{apmf} code transmits a $K$-bit message over $T$ rounds. The first forward transmission follows \gls{npmf}, $X_1 = \frac{\lambda_1}{D_0} \Theta$. For $i=1,\dots,T-1$, after observing $\mathbf{Y}^i$, the receiver computes the posterior mean $\widehat{\Theta}_i^{\rm RX}
    =
    \mathbb{E}[\Theta\mid \mathbf{Y}^i]$ and sends the normalized feedback signal
\begin{equation}
    \widetilde{X}_i
    =
    \frac{\omega_i}{\nu_i}
    \widehat{\Theta}_i^{\rm RX},
    \qquad
    \nu_i
    =
    \sqrt{\mathbb{E}\!\left[(\widehat{\Theta}_i^{\rm RX})^2\right]}. 
\end{equation}
The transmitter observes $\widetilde{Y}_i=\widetilde{X}_i+\widetilde{N}_i$ and forms
\begin{equation*}
    \widehat{\Theta}_i^{\rm TX}
    =
    \frac{\widetilde{Y}_i}{\omega_i/\nu_i}.
\end{equation*}
For $i=2,\dots,T$, the transmitter sends
\begin{equation*}
    X_i
    =
    \frac{\lambda_i}{D_{i-1}}
    \left(
    \Theta-\widehat{\Theta}_{i-1}^{\rm TX}
    \right),
    \quad
    D_{i-1}
    =
    \sqrt{
    \mathbb{E}
    \left[
    \left(
    \Theta-\widehat{\Theta}_{i-1}^{\rm TX}
    \right)^2
    \right]
    }.
\end{equation*}
Here, $\lambda_i$ and $\omega_i$ denote the forward and feedback power allocation coefficients, respectively. Under the per-codeword average power constraint, they satisfy
\[
    \sum_{i=1}^{T}\lambda_i^2 = T,
    \qquad
    \sum_{i=1}^{T-1}\omega_i^2 = T-1.
\]
\end{definition}

The receiver updates the accumulated log-likelihood used for \gls{map} decoding. For the first round, we set $\widehat{\Theta}_0^{\rm RX}=0$ and $\sigma_{{\rm eff},1}^2=\sigma_f^2$. For $i\geq 2$, the transmitter uses the noisy estimate
\begin{equation*}
    \widehat{\Theta}_{i-1}^{\rm TX}
    =
    \widehat{\Theta}_{i-1}^{\rm RX}
    +
    \frac{\widetilde{N}_{i-1}}{\omega_{i-1}/\nu_{i-1}} .
\end{equation*}
Substituting this into the forward transmission gives
\begin{equation*}
    Y_i
    =
    \frac{\lambda_i}{D_{i-1}}
    \left(
    \Theta-\widehat{\Theta}_{i-1}^{\rm RX}
    \right)
    +
    N_i
    -
    \frac{\lambda_i/D_{i-1}}{\omega_{i-1}/\nu_{i-1}}
    \widetilde{N}_{i-1}.
\end{equation*}
Thus, for $i\geq 2$, the receiver uses the effective variance
\begin{equation}
\label{eq:sigma_eff_apmf}
    \sigma_{{\rm eff},i}^2
    =
    \sigma_f^2
    +
    \left(
    \frac{\lambda_i/D_{i-1}}{\omega_{i-1}/\nu_{i-1}}
    \right)^2
    \sigma_{fb}^2.
\end{equation} 

For each hypothesis $\theta_j$, the accumulated log-likelihood at the receiver is updated as
\begin{equation}
\label{eq:apmf_ll_update}
    \ell_i(\theta_j)
    =
    \ell_{i-1}(\theta_j)
    -
    \frac{
    \left(
    Y_i
    -
    \frac{\lambda_i}{D_{i-1}}
    \left(
    \theta_j-\widehat{\Theta}_{i-1}^{\rm RX}
    \right)
    \right)^2
    }
    {2\sigma_{{\rm eff},i}^2}.
\end{equation}
The posterior distribution is obtained by applying the softmax function to the accumulated log-likelihoods. After $T$ rounds, the receiver applies the \gls{map} rule
\begin{equation*}
    \widehat{\Theta}
    =
    \arg\max_{\theta_j\in\{\theta_1,\dots,\theta_U\}}
    \mathbb{P}(\Theta=\theta_j\mid \mathbf{Y}^T),
\end{equation*}
which is then mapped to the decoded message $\widehat{\mathbf{M}}$.

The error analysis follows the same steps as in the
noiseless-feedback \gls{npmf} analysis, with the forward noise variance replaced by the round-dependent effective variance $\sigma_{{\rm eff},i}^2$. The detailed derivation is given in Appendix~\ref{app:apmf_error_analysis}. The resulting probability of error is bounded
\begin{equation}
    \BLER_{\text{PAM-APMF}}
    \leq
    \frac{2}{U}
    \sum_{m=1}^{U-1}
    (U-m)
    Q\!\left(
    m\sqrt{\widetilde{\Lambda}_T}
    \right),
\end{equation}
where
\begin{equation}
    \widetilde{\Lambda}_T
    =
    \eta^2
    \sum_{i=1}^{T}
    \frac{
    \left(\lambda_i/D_{i-1}\right)^2
    }
    {\sigma_{{\rm eff},i}^2}.
\end{equation}

The bound highlights a limitation of the \gls{pam}-based design. The accumulated effective distance is scaled by the \gls{pam} spacing parameter $\eta=\sqrt{3/(2^{2K}-1)}$, which decreases as $K$ increases. Hence, at a fixed rate, the $2^K$ \gls{pam} constellation becomes increasingly dense in one dimension and more sensitive to feedback-induced mismatch. This motivates a projection-based \gls{apmf} code, which represents messages using higher-dimensional binary vectors and refines learned scalar projections instead of a single \gls{pam} coordinate. Unlike one-dimensional \gls{pam}, where all $2^K$ messages are placed on a fixed constellation, the projection-based scheme allows different projection directions to be used in different transmission rounds. As a result, the distances between message pairs can change over time, allowing the encoder to gradually separate different messages instead of relying on a single fixed constellation, thus reducing the packing loss of one-dimensional signaling.

\subsection{Projection-Based \gls{apmf} Code}
To avoid the one-dimensional packing limitation of \gls{pam}, we keep the vector representation of the message and use learned scalar projections for posterior-mean refinement.

\begin{definition}[Projection-Based Active Posterior-Mean Feedback Code]
The message $\mathbf{M}\in\F_2^K$ is mapped to its BPSK representation
\begin{equation*}
    \mathbf{B}
    =
    2\mathbf{M}-\mathbf{1}
    \in \{-1,+1\}^K .
\end{equation*}
For each round $i$, let $\mathbf{a}_i\in\R^K$ be a projection vector satisfying $\|\mathbf{a}_i\|_2=1$. The corresponding projected feature is
\begin{equation}
    \Phi_i
    =
    \mathbf{a}_i^\top \mathbf{B}.
\end{equation}
The projection vectors $\{\mathbf{a}_i\}_{i=1}^T$ are learned code parameters fixed during transmission and known to both the transmitter and receiver. For each $\mathbf{m}\in\F_2^K$, define
\begin{equation*}
    \mathbf{b}(\mathbf{m})
    =
    2\mathbf{m}-\mathbf{1},
    \qquad
    \phi_i(\mathbf{m})
    =
    \mathbf{a}_i^\top \mathbf{b}(\mathbf{m}).
\end{equation*}

In the first round, the transmitter sends $X_1
    =
    \frac{\lambda_1}{D_0}\Phi_1$ with $D_0
    =
    \sqrt{\E[\Phi_1^2]}$. For $i=2,\dots,T$, the feedback encoder computes the posterior probabilities based on $\mathbf{Y}^{i-1}$ and then evaluates the posterior mean of the next projected feature,
\begin{equation*}
    \widehat{\Phi}_{i}^{\rm RX}
    =
    \E[\Phi_i\mid \mathbf{Y}^{i-1}]
    =
    \sum_{\mathbf{m}\in\F_2^K}
    \phi_i(\mathbf{m})
    \Prob(\mathbf{M}=\mathbf{m}\mid \mathbf{Y}^{i-1}),
\end{equation*}
and sends the normalized feedback signal
\begin{equation*}
    \widetilde{X}_{i-1}
    =
    \frac{\omega_{i-1}}{\nu_{i-1}}
    \widehat{\Phi}_{i}^{\rm RX},
    \qquad
    \nu_{i-1}
    =
    \sqrt{\E\!\left[
    \left(\widehat{\Phi}_{i}^{\rm RX}\right)^2
    \right]} .
\end{equation*}
The transmitter forms the estimate $\widehat{\Phi}_{i}^{\rm TX}
    =
    \frac{\widetilde{Y}_{i-1}}{\omega_{i-1}/\nu_{i-1}}$ and then sends the normalized residual
\begin{equation*}
    X_i
    =
    \frac{\lambda_i}{D_{i-1}}
    \left(
    \Phi_i-\widehat{\Phi}_{i}^{\rm TX}
    \right),
    \quad
    D_{i-1}
    =
    \sqrt{
    \E\!\left[
    \left(
    \Phi_i-\widehat{\Phi}_{i}^{\rm TX}
    \right)^2
    \right]
    } .
\end{equation*}
Here, $\lambda_i$ and $\omega_i$ denote the forward and feedback power allocation coefficients, respectively.
\end{definition}

We set $\widehat{\Phi}_{1}^{\rm RX}=0$ and $\sigma_{{\rm eff},1}^2=\sigma_f^2$. For $i\geq 2$, as in the \gls{pam}-based \gls{apmf} code, the effective variance is 
\begin{equation*}
\sigma_{{\rm eff},i}^2
=
\sigma_f^2
+
\left(
\frac{\lambda_i/D_{i-1}}
{\omega_{i-1}/\nu_{i-1}}
\right)^2
\sigma_{fb}^2.
\end{equation*}

For each hypothesis $\mathbf{m}\in\F_2^K$, the accumulated log-likelihood at the receiver is updated as
\begin{equation*}
\label{eq:projection_apmf_ll_update}
    \ell_i(\mathbf{m})
    =
    \ell_{i-1}(\mathbf{m})
    -
    \frac{
    \left(
    Y_i
    -
    \frac{\lambda_i}{D_{i-1}}
    \left(
    \phi_i(\mathbf{m})-\widehat{\Phi}_{i}^{\rm RX}
    \right)
    \right)^2
    }
    {2\sigma_{{\rm eff},i}^2}.
\end{equation*}
After $T$ rounds, the receiver applies the \gls{map} decoder $\widehat{\mathbf{M}}
    =
    \arg\max_{\mathbf{m}\in\F_2^K}
    \ell_T(\mathbf{m})$.

For fixed projection vectors $\{\mathbf{a}_i\}_{i=1}^T$, the pairwise distance is no longer determined by the one-dimensional \gls{pam} spacing. For two messages $\mathbf{m},\mathbf{m}'\in\F_2^K$, the round-$i$ distance is
\begin{equation*}
    \Delta_i(\mathbf{m},\mathbf{m}')
    =
    \mathbf{a}_i^\top
    \left(
    \mathbf{b}(\mathbf{m})-\mathbf{b}(\mathbf{m}')
    \right).
\end{equation*}
As derived in Appendix~\ref{app:projection_apmf_error_analysis}, the resulting union bound is
\begin{equation*}
    \BLER_{\text{Proj-APMF}}
    \leq
    \frac{1}{2^K}
    \sum_{\mathbf{m}\in\F_2^K}
    \sum_{\substack{\mathbf{m}'\in\F_2^K\\ \mathbf{m}'\neq \mathbf{m}}}
    Q\!\left(
    \sqrt{
    \widetilde{\Lambda}_T^{\rm Proj}(\mathbf{m},\mathbf{m}')
    }
    \right).
\end{equation*}
where \begin{equation*}
    \widetilde{\Lambda}_T^{\rm Proj}(\mathbf{m},\mathbf{m}')
    =
    \frac{1}{4}
    \sum_{i=1}^{T}
    \frac{
    \left(\lambda_i/D_{i-1}\right)^2
    \Delta_i^2(\mathbf{m},\mathbf{m}')
    }
    {\sigma_{{\rm eff},i}^2}.
\end{equation*}

Unlike the \gls{pam}-based bound, this bound cannot be grouped only by $|j-k|$, since the pairwise distance depends on both the message pair and the projection vector in each round. This gives the projection-based \gls{apmf} code additional flexibility beyond the fixed one-dimensional \gls{pam} spacing. In fact, the \gls{pam}-based \gls{apmf} code is a special case of the projection-based scheme. If every projection vector is fixed as
\begin{equation*}
    \mathbf{a}_i
    =
    \frac{
    [2^{K-1},2^{K-2},\dots,1]^\top
    }{
    \sqrt{\sum_{j=0}^{K-1}4^j}
    },
    \qquad i=1,\dots,T,
\end{equation*}
then $\mathbf{a}_i^\top\mathbf{B}$ exactly recovers the $2^K$-\gls{pam} constellation. Therefore, the \gls{pam}-based \gls{apmf} code is recovered by using this fixed projection direction in every round.

\subsection{Simulation Results}
For learned projection-based \gls{apmf}, the projection vectors are randomly initialized and jointly optimized with the power allocation coefficients. Under the per-codeword average power constraint, the model has only $TK+(2T-1)$ trainable parameters from the projection vectors and power allocation coefficients. Each setting is trained for $5\times 10^4$ batches and evaluated over three independent random seeds. After training, the normalization constants $D_i$ and $\nu_i$ are calibrated using a large independent batch and then fixed during testing.

\begin{figure}[!t]
    \centering
    \includegraphics[width=0.9\columnwidth]{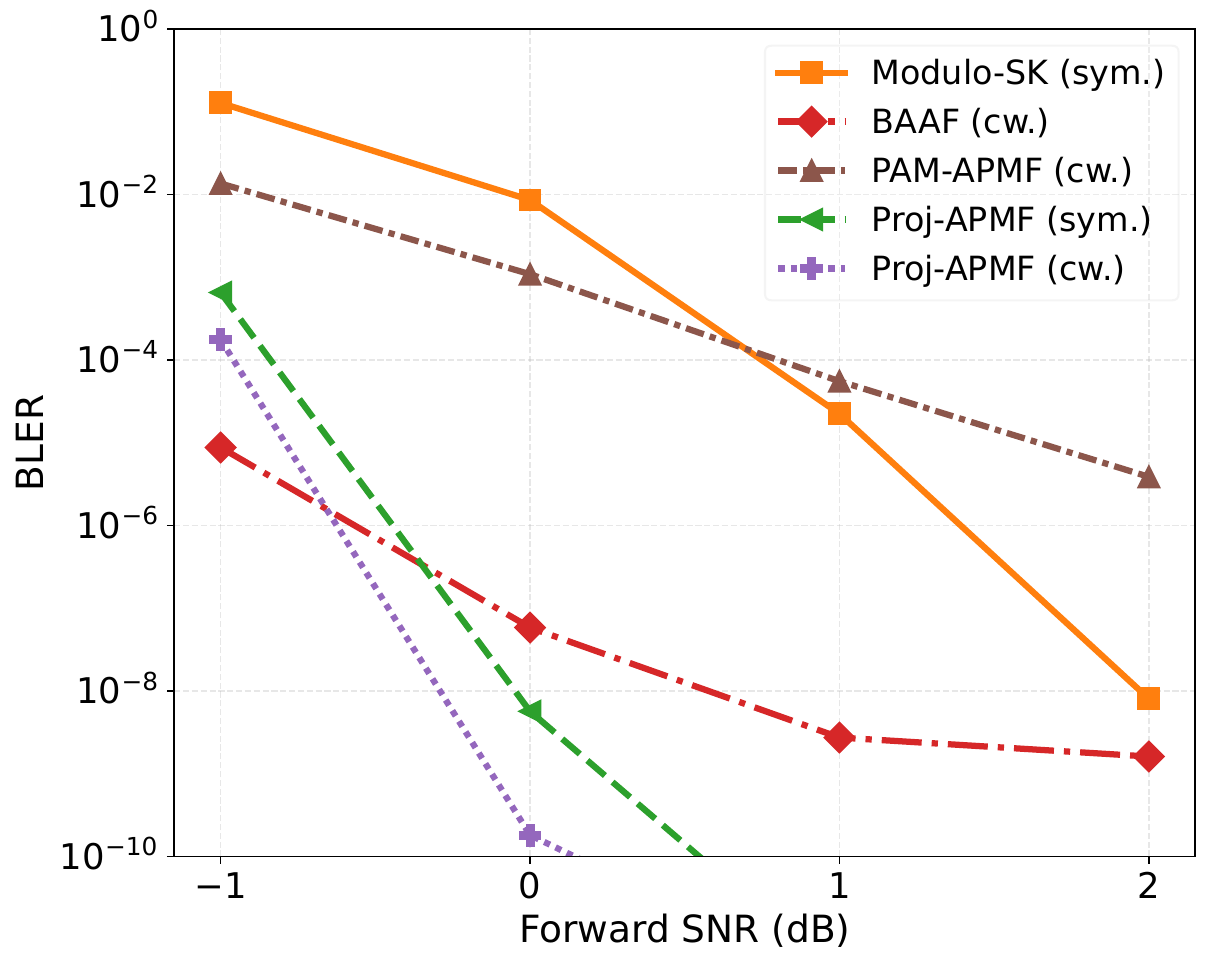}
    \caption{Performance of existing active feedback schemes for $K/T = 3/9$ with feedback SNR fixed at 20 dB.}
    \label{fig:apmf_noisy_performance}
\end{figure}
  
As shown in Fig.~\ref{fig:apmf_noisy_performance}, we compare existing active feedback schemes for $K=3$ and $T=9$ with the feedback SNR fixed at $20$ dB. The projection-based \gls{apmf} under the per-codeword average power constraint achieves the best BLER for $\mathrm{SNR}_f \geq 0$ dB, followed by its per-symbol average-power version. Both outperform BAAF in this region with far fewer learned parameters and an interpretable posterior-mean structure. At $\mathrm{SNR}_f=-1$ dB, BAAF performs best, while the projection-based \gls{apmf} remains competitive. The PAM-based \gls{apmf} is less effective than the projection-based design, especially under the per-codeword average power constraint, but it still outperforms Modulo-SK at low SNR. As the forward SNR increases, Modulo-SK improves rapidly, but it is still outperformed by the proposed projection-based \gls{apmf} schemes for the same message size.


To further study the effect of feedback, Fig.~\ref{fig:activeFB}
fixes the forward SNR at $0$ dB and varies the feedback SNR.
All schemes are evaluated at the same code rate, $R=K/T=1/3$.
For $K=1$, both ActiveFB and the projection-based \gls{apmf}
outperform Modulo-SK, with ActiveFB achieving the lowest
\gls{bler}. However, extending ActiveFB beyond $K=2$ is not
straightforward~\cite{activeFB}, {as its recurrent architecture is designed for very short messages,} while the projection-based
\gls{apmf} naturally supports larger message lengths. The projection-based \gls{apmf} also improves significantly as
the feedback SNR increases, with $K=3$ achieving the lowest
\gls{bler} at high feedback SNR.

\begin{figure}[!t]
    \centering
    \includegraphics[width=0.9\columnwidth]{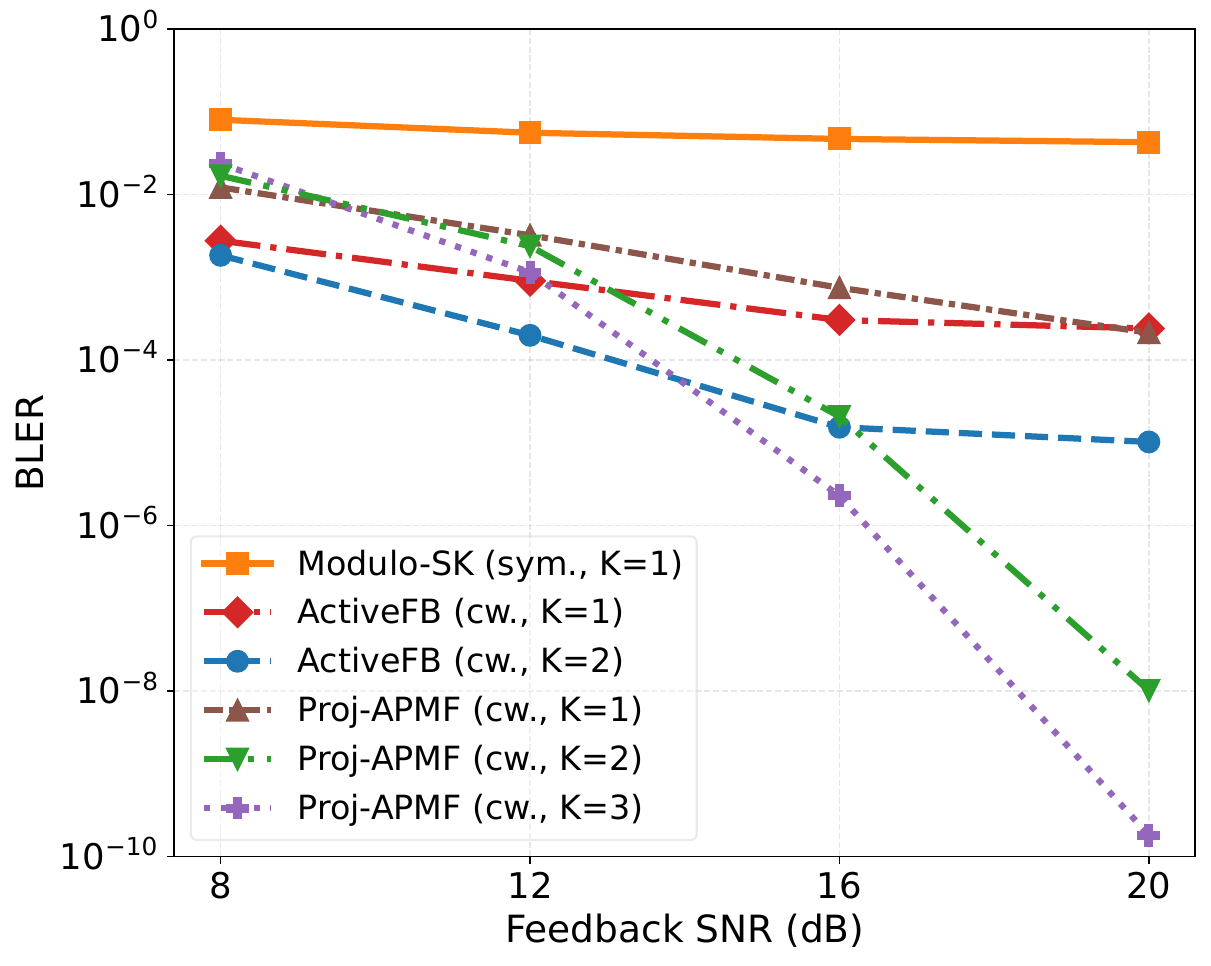}
    \caption{{Performance of existing active feedback schemes with $R=1/3$ and forward SNR fixed at $0$ dB.}}
    \label{fig:activeFB}
\end{figure}


Finally, we examine the effect of the message length $K$ in Table~\ref{tab:k_effect}. We fix the code rate at $R=1/3$, the forward SNR at $0$ dB, and the feedback SNR at $20$ dB. 
For $K=1$, the \gls{pam}-based and projection-based \gls{apmf} schemes are equivalent, since the projection reduces to a binary constellation. As $K$ increases, the \gls{pam}-based scheme becomes less reliable because more messages are packed into a one-dimensional constellation, making the posterior update more sensitive to feedback noise, especially for $K=3$ and $K=4$. In contrast, the projection-based scheme refines the message over multiple dimensions and remains robust as $K$ increases, achieving BLER below $10^{-8}$ for $K\geq 2$. These results show that the projection-based design is essential for extending \gls{apmf} to larger message sizes under noisy feedback.

\begin{table}[!t]
\centering
\caption{BLER of \gls{pam}-based and projection-based \gls{apmf} for different $(K,T)$}
\label{tab:k_effect}
\begin{tabular}{c|c|c}
\hline
$(K,T)$ & PAM-based \gls{apmf} & Projection-based \gls{apmf} \\
\hline
$(1,3)$  & $2.14\times10^{-4}$ & $2.18\times10^{-4}$ \\
$(2,6)$  & $2.86\times10^{-6}$ & $3.25\times10^{-9}$ \\
$(3,9)$  & $1.07\times10^{-3}$ & $1.82\times10^{-10}$ \\
$(4,12)$ & $3.65\times10^{-2}$ & $3.94\times10^{-11}$ \\
\hline
\end{tabular}
\end{table}

\begin{figure}[!t]
    \centering
    \includegraphics[width=0.95\columnwidth]{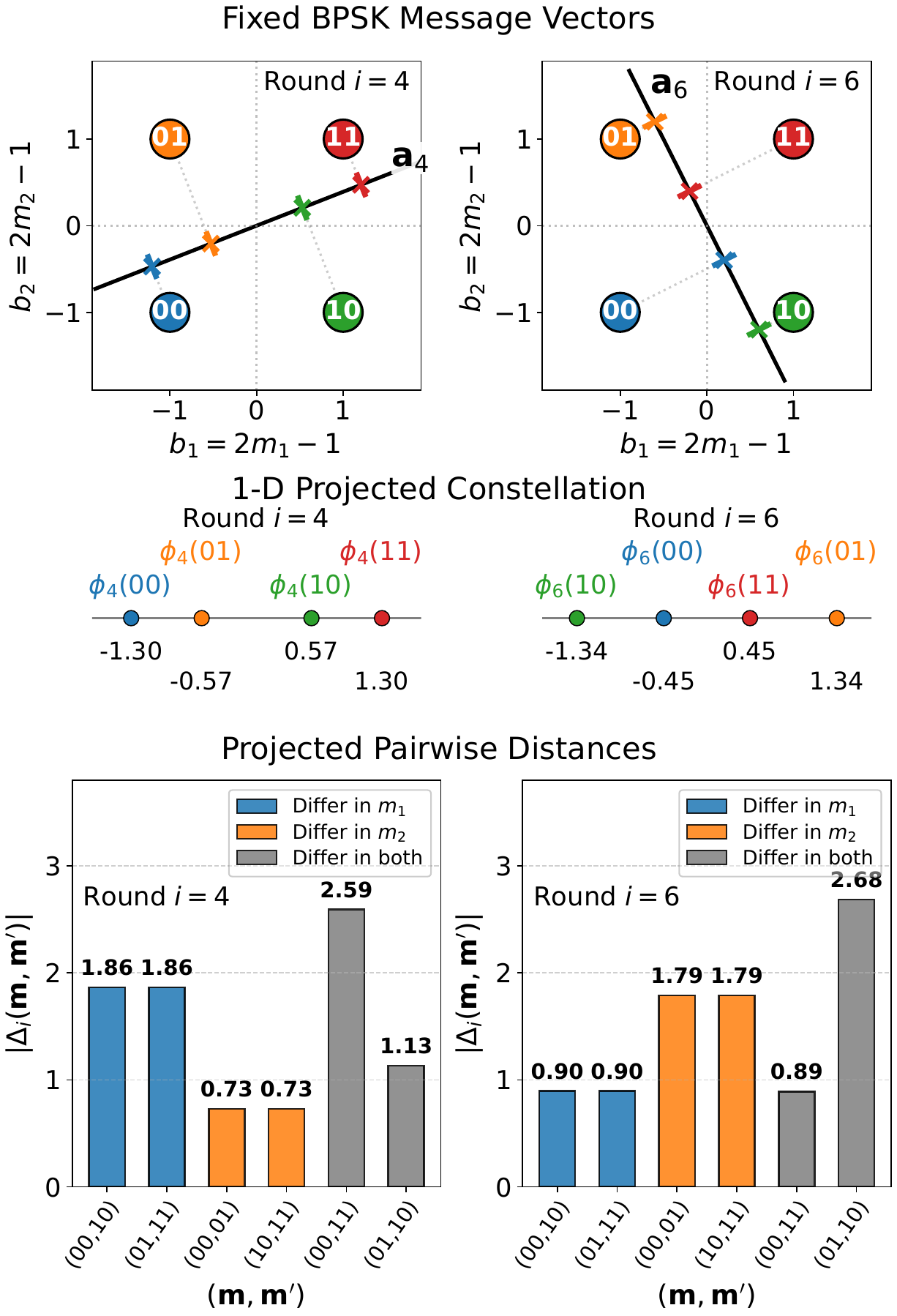}
   \caption{Learned projections at rounds $i=4$ and $i=6$
    for $K=2$ and $T=6$, with forward SNR $0$ dB and feedback SNR
    $20$ dB. The top row shows the BPSK representations $\mathbf{b}(\mathbf{m})=(b_1,b_2)$ of the four messages. The black line denotes the learned projection direction
    $\mathbf{a}_i$ for the indicated round.
    The middle row shows the corresponding one-dimensional projected
    constellation
    $\phi_i(\mathbf{m})
    =\mathbf{a}_i^\top\mathbf{b}(\mathbf{m})$,
    and the bottom row shows the projected pairwise distances
    $|\Delta_i(\mathbf{m},\mathbf{m}')|
    =|\phi_i(\mathbf{m})-\phi_i(\mathbf{m}')|$.
    The two displayed rounds illustrate that different learned
    projection directions emphasize the separation of different
    message pairs.
    Messages are labeled by the binary string $m_1m_2$.}
   \label{fig:projection_interpretation}
\end{figure}

To provide intuition for the learned projection-based \gls{apmf} code,
Fig.~\ref{fig:projection_interpretation} shows two representative
rounds, $i=4$ and $i=6$, for $K=2$ and $T=6$, with forward SNR
$0$ dB and feedback SNR $20$ dB.
Unlike the \gls{pam}-based scheme, which uses the same projection
direction in every round, the proposed scheme learns a separate
projection vector $\mathbf{a}_i$ for each round.
At round $i=4$, $\mathbf{a}_4$ is nearly aligned with the $b_1$ axis.
Thus, message pairs that differ in $m_1$ have larger projected
distances than pairs that differ in $m_2$, and $(00,11)$ has the
largest separation.
At round $i=6$, $\mathbf{a}_6$ is nearly aligned with the $b_2$ axis.
Thus, pairs that differ in $m_2$ have larger projected distances, and
$(01,10)$ has the largest separation.
These examples show that different learned projection directions can
emphasize different message pairs across rounds.
Since the likelihood update depends on
$\Delta_i(\mathbf{m},\mathbf{m}')$, larger projected distances make
the corresponding message hypotheses easier to distinguish.
The variation of $\mathbf{a}_i$ is learned rather than imposed by
design.
A theoretical explanation of this round-dependent behavior remains
an interesting open problem.

\section{Conclusion}\label{sec:conclusion}

{In this paper, we develop a posterior-mean feedback coding framework
for short-blocklength AWGN channels with feedback. For noiseless passive
feedback, the proposed \gls{npmf} uses a closed-form encoder and optimal \gls{map} decoding, reduces to
\gls{ppv} for $K=1$, and naturally extends to multi-bit messages.
In the considered short-blocklength settings, \gls{npmf} achieves the
lowest \gls{bler} among the evaluated analytical and learned baselines.
For noisy active feedback, we propose \gls{apmf} and a learned
projection-based extension that mitigates the one-dimensional packing
limitation of \gls{pam} and improves robustness for larger message sizes. The projection-based scheme preserves an interpretable structure with only a small number of trainable parameters, while achieving competitive or better finite-blocklength performance. \footnote{The manuscript was written by the authors. OpenAI Codex 5.5
was used for grammar checking and to assist in verifying the lower-bound
derivation in Appendix~\ref{apx:sandwich}.}}

\bibliographystyle{IEEEtran}
\bibliography{reference}

\appendices
\section{Proof of Proposition~\ref{prop:sandwich}}
\label{apx:sandwich}

\begin{IEEEproof}
Let $p_{\mathbf{m}} \coloneqq \mathbb{P}(\mathbf{M}=\mathbf{m}\mid\mathbf{Y}^T)$,
$\mathbf{p} \coloneqq \{p_{\mathbf{m}}\}_{\mathbf{m} \in \mathbb{F}_2^K}$ denote
the posterior distribution over all $2^K$ messages given $\mathbf{Y}^T$, and
$p_{\max} \coloneqq \max_{\mathbf{m}} p_{\mathbf{m}}$. Define the pointwise
\gls{map} error probability
\begin{equation}
    e \;\coloneqq\; 1 - p_{\max},
    \label{eq:e_def}
\end{equation}
so that $\BLER = \mathbb{E}[e]$ and, by \eqref{eq:ce_min},
$\Loss_{\mathrm{CCE}}^{\star} = \mathbb{H}(\mathbf{M}\mid\mathbf{Y}^T)
= \mathbb{E}[\mathbb{H}(\mathbf{p})]$, where $\mathbb{H}(\mathbf{p})
= -\sum_{\mathbf{m}} p_{\mathbf{m}} \log_2 p_{\mathbf{m}}$ is the entropy
of the posterior for a specific channel realization. Both bounds are derived
by first establishing a pointwise inequality on $\mathbb{H}(\mathbf{p})$,
then taking the expectation over all channel realizations.

\textbf{Key decomposition.}
If $e=0$, then $\mathbb{H}(\mathbf{p})=0$, and both pointwise
bounds hold with equality. For $e>0$, choose
$\widehat{\mathbf{m}}=\arg\max_{\mathbf{m}}p_{\mathbf{m}}$
and define $r_{\mathbf{m}}\coloneqq p_{\mathbf{m}}/e$ for
$\mathbf{m}\neq\widehat{\mathbf{m}}$. Then
$\mathbf{r}=\{r_{\mathbf{m}}\}_{\mathbf{m}\neq\widehat{\mathbf{m}}}$
is a probability distribution over the remaining $2^K-1$ messages.

Substituting $p_{\widehat{\mathbf{m}}}=1-e$ and $p_{\mathbf{m}}=e\,r_{\mathbf{m}}$
into the entropy definition and expanding
$\log_2(e\,r_{\mathbf{m}})=\log_2 e+\log_2 r_{\mathbf{m}}$ yields
\begin{equation}
    \mathbb{H}(\mathbf{p})
    \;=\;
    \mathbb{H}_{2}(e)
    \;+\;
    e\,\mathbb{H}(\mathbf{r}),
    \label{eq:decomp}
\end{equation}
where
$\mathbb{H}(\mathbf{r}) = -\sum_{\mathbf{m}\neq\widehat{\mathbf{m}}}
 r_{\mathbf{m}}\log_2 r_{\mathbf{m}} \geq 0$
is the conditional entropy over incorrect messages.

\textbf{Lower bound.}
We show that $\mathbb{H}(\mathbf{p})\geq 2e$.
For $0< e\leq 1/2$, \eqref{eq:decomp} gives
$\mathbb{H}(\mathbf{p})\geq\mathbb{H}_{2}(e)$.
Since $\mathbb{H}_{2}$ is concave with
$\mathbb{H}_{2}(0)=0$ and $\mathbb{H}_{2}(1/2)=1$,
we have $\mathbb{H}_{2}(e)\geq 2e$ on this interval.

For $e>1/2$, every posterior probability satisfies
$p_{\mathbf{m}}\leq p_{\max}=1-e$, so
\[
\mathbb{H}(\mathbf{p})\geq -\log_2(1-e).
\]
Let $g(e)=-\log_2(1-e)$. By convexity,
\[
g(e)\geq
1+\frac{2}{\ln 2}\left(e-\frac12\right)
\geq 2e,
\]
where we used $g(1/2)=1$, $g'(1/2)=2/\ln 2>2$,
and $e>1/2$.

Thus, $\mathbb{H}(\mathbf{p})\geq 2e$ in both cases.
Taking expectations yields
\begin{equation}
\Loss_{\mathrm{CCE}}^{\star}
=
\mathbb{E}[\mathbb{H}(\mathbf{p})]
\geq
2\mathbb{E}[e]
=
2\BLER.
\end{equation}

\textbf{Upper bound.}
Since $\mathbb{H}(\mathbf{r})\leq\log_2(2^{K}-1)$,
equation~\eqref{eq:decomp} gives
$\mathbb{H}(\mathbf{p})\leq\mathbb{H}_{2}(e)+e\log_2(2^{K}-1)$.
Taking expectations and applying Jensen's inequality to the concave
$\mathbb{H}_{2}$:
\begin{align}
    \Loss_{\mathrm{CCE}}^{\star}
    &\leq
    \mathbb{E}[\mathbb{H}_{2}(e)]
    +\mathbb{E}[e]\log_2(2^{K}-1)
    \nonumber\\
    &\leq
    \mathbb{H}_{2}(\mathbb{E}[e])
    +\mathbb{E}[e]\log_2(2^{K}-1)
    \nonumber\\
    &=
    \mathbb{H}_{2}(\BLER)
    +\BLER\log_2(2^{K}-1).
\end{align}
\end{IEEEproof}

\section{Error Analysis for \gls{pam}-Based \gls{apmf}}
\label{app:apmf_error_analysis}
We derive a union bound on the BLER of the \gls{pam}-based \gls{apmf} code. The effective variance $\sigma_{{\rm eff},i}^2 = \sigma_f^2
    +
    \left(
    \frac{\lambda_i/D_{i-1}}{\omega_{i-1}/\nu_{i-1}}
    \right)^2
    \sigma_{fb}^2$ is defined in \eqref{eq:sigma_eff_apmf}. Suppose the true transmitted symbol is $\theta_j$. A pairwise error in favor of $\theta_k$ occurs if
\begin{equation*}
    S_T^{(j,k)}
    =
    \ell_T(\theta_j)-\ell_T(\theta_k)
    <0 .
\end{equation*}
From the log-likelihood update in \eqref{eq:apmf_ll_update}, under the true hypothesis $\theta_j$,
\begin{equation*}
    Y_i
    =
    \frac{\lambda_i}{D_{i-1}}
    \left(
    \theta_j-\widehat{\Theta}_{i-1}^{\rm RX}
    \right)
    +
    N_{{\rm eff},i},
\end{equation*}
where $D_0=1$ and $N_{{\rm eff},i}$ has variance $\sigma_{{\rm eff},i}^2$. Therefore,
\begin{equation*}
\begin{aligned}
S_T^{(j,k)}
=
\sum_{i=1}^{T}\Bigg[
&\frac{
(\lambda_i/D_{i-1})^2(\theta_j-\theta_k)^2
}{
2\sigma_{{\rm eff},i}^2
} \\
&+
\frac{
(\lambda_i/D_{i-1})(\theta_j-\theta_k)
}{
\sigma_{{\rm eff},i}^2
}
N_{{\rm eff},i}
\Bigg].
\end{aligned}
\end{equation*}
Hence, $S_T^{(j,k)}$ is a Gaussian random variable with mean
\begin{equation*}
    \mu_{j,k}
    =
    \sum_{i=1}^{T}
    \frac{
    \left(\lambda_i/D_{i-1}\right)^2
    (\theta_j-\theta_k)^2
    }
    {2\sigma_{{\rm eff},i}^2}
\end{equation*}
and variance
\begin{equation*}
    \sigma_{j,k}^2
    =
    \sum_{i=1}^{T}
    \frac{
    \left(\lambda_i/D_{i-1}\right)^2
    (\theta_j-\theta_k)^2
    }
    {\sigma_{{\rm eff},i}^2}
    =
    2\mu_{j,k}.
\end{equation*}
Since $\theta_j-\theta_k=2\eta(j-k)$, the pairwise error probability is
\begin{equation*}
    \mathbb{P}\!\left(S_T^{(j,k)}<0\right)
    =
    Q\!\left(
    |j-k|\sqrt{\widetilde{\Lambda}_T}
    \right),
\end{equation*}
where $\widetilde{\Lambda}_T
    =
    \eta^2
    \sum_{i=1}^{T}
    \frac{
    \left(\lambda_i/D_{i-1}\right)^2
    }
    {\sigma_{{\rm eff},i}^2}$.

Applying the union bound over all \gls{pam} symbol pairs gives
\begin{equation}
    \BLER_{\text{PAM-APMF}}
    \leq
    \frac{2}{U}
    \sum_{m=1}^{U-1}
    (U-m)
    Q\!\left(
    m\sqrt{\widetilde{\Lambda}_T}
    \right).
\end{equation}

\section{Error Analysis for Projection-Based \gls{apmf}}
\label{app:projection_apmf_error_analysis}
We derive a union bound of the projection-based \gls{apmf} code for fixed projection vectors $\{\mathbf{a}_i\}_{i=1}^T$. Suppose the true message is $\mathbf{m}\in\F_2^K$. An error occurs if there exists another message $\mathbf{m}'\neq \mathbf{m}$ such that
\begin{equation*}
    S_T^{(\mathbf{m},\mathbf{m}')}
    =
    \ell_T(\mathbf{m})-\ell_T(\mathbf{m}')
    <0 .
\end{equation*}

For each round $i$, define the projected pairwise distance
\begin{equation}
    \Delta_i(\mathbf{m},\mathbf{m}')
    =
    \phi_i(\mathbf{m})-\phi_i(\mathbf{m}')
    =
    \mathbf{a}_i^\top
    \left(
    \mathbf{b}(\mathbf{m})-\mathbf{b}(\mathbf{m}')
    \right).
\end{equation}
From the log-likelihood update, under the true message $\mathbf{m}$,
\begin{equation*}
    Y_i
    =
    \frac{\lambda_i}{D_{i-1}}
    \left(
    \phi_i(\mathbf{m})-\widehat{\Phi}_{i}^{\rm RX}
    \right)
    +
    N_{{\rm eff},i},
\end{equation*}
where $D_0=\sqrt{\E[\Phi_1^2]}$ and $N_{{\rm eff},i}$ has variance $\sigma_{{\rm eff},i}^2$. Therefore,
\begin{equation*}
\begin{aligned}
S_T^{(\mathbf{m},\mathbf{m}')}
=
\sum_{i=1}^{T}\Bigg[
&\frac{
(\lambda_i/D_{i-1})^2
\Delta_i^2(\mathbf{m},\mathbf{m}')
}{
2\sigma_{{\rm eff},i}^2
} \\
&+
\frac{
(\lambda_i/D_{i-1})
\Delta_i(\mathbf{m},\mathbf{m}')
}{
\sigma_{{\rm eff},i}^2
}
N_{{\rm eff},i}
\Bigg].
\end{aligned}
\end{equation*}
Hence, $S_T^{(\mathbf{m},\mathbf{m}')}$ is a Gaussian random variable with mean
\begin{equation*}
    \mu_{\mathbf{m},\mathbf{m}'}
    =
    \sum_{i=1}^{T}
    \frac{
    \left(\lambda_i/D_{i-1}\right)^2
    \Delta_i^2(\mathbf{m},\mathbf{m}')
    }
    {2\sigma_{{\rm eff},i}^2}
\end{equation*}
and variance
\begin{equation*}
    \sigma_{\mathbf{m},\mathbf{m}'}^2
    =
    \sum_{i=1}^{T}
    \frac{
    \left(\lambda_i/D_{i-1}\right)^2
    \Delta_i^2(\mathbf{m},\mathbf{m}')
    }
    {\sigma_{{\rm eff},i}^2}
    =
    2\mu_{\mathbf{m},\mathbf{m}'}.
\end{equation*}
Thus,
\begin{equation*}
    \Prob\!\left(
    S_T^{(\mathbf{m},\mathbf{m}')}<0
    \right)
    =
    Q\!\left(
    \sqrt{
    \widetilde{\Lambda}_T^{\rm Proj}(\mathbf{m},\mathbf{m}')
    }
    \right),
\end{equation*}
where
\begin{equation*}
    \widetilde{\Lambda}_T^{\rm Proj}(\mathbf{m},\mathbf{m}')
    =
    \frac{1}{4}
    \sum_{i=1}^{T}
    \frac{
    \left(\lambda_i/D_{i-1}\right)^2
    \Delta_i^2(\mathbf{m},\mathbf{m}')
    }
    {\sigma_{{\rm eff},i}^2}.
\end{equation*}
Applying the union bound over all message pairs gives
\begin{equation}
    \BLER_{\text{Proj-APMF}}
    \leq
    \frac{1}{2^K}
    \sum_{\mathbf{m}\in\F_2^K}
    \sum_{\substack{\mathbf{m}'\in\F_2^K\\ \mathbf{m}'\neq \mathbf{m}}}
    Q\!\left(
    \sqrt{
    \widetilde{\Lambda}_T^{\rm Proj}(\mathbf{m},\mathbf{m}')
    }
    \right).  
\end{equation}



 




\end{document}